\documentclass[journal, twocolumn]{IEEEtran}
\usepackage[english]{babel}
\usepackage{amsmath,amssymb,bm}
\usepackage{amsfonts}
\usepackage{graphicx}
\usepackage{subfigure}
\usepackage{epstopdf}
\usepackage{subeqnarray}
\usepackage{color}

\newcommand{\qh}{{\bf h}}

\newcommand{\qw}{{\bf w}}
\newcommand{\qx}{{\bf x}}

\newcommand{\qA}{{\bf A}}

\newcommand{\qD}{{\bf D}}

\newcommand{\qH}{{\bf H}}
\newcommand{\qI}{{\bf I}}

\newcommand{\qQ}{{\bf Q}}

\newcommand{\qU}{{\bf U}}

\newcommand{\qW}{{\bf W}}
\newcommand{\qX}{{\bf X}}

\newcommand{\qzero}{{\bf 0}}

\newcommand{\tr}{\mbox{trace}}
 
\newtheorem{theorem}{Theorem}
\newtheorem{lemma}{Lemma}

\newcommand{\be}{\begin{equation}} \newcommand{\ee}{\end{equation}}
\newcommand{\bea}{\begin{eqnarray}} \newcommand{\eea}{\end{eqnarray}}

\begin{document}

\title{Specific Absorption Rate-Aware  Beamforming   in MISO Downlink SWIPT Systems}%
\author{Juping Zhang,
        Gan Zheng,~\IEEEmembership{Senior~Member,~IEEE,}
         Ioannis Krikidis,~\IEEEmembership{Fellow,~IEEE,}
       and Rui Zhang,~\IEEEmembership{Fellow,~IEEE}
  \thanks{J. Zhang,  and G. Zheng are  with the Wolfson School of Mechanical, Electrical and Manufacturing Engineering, Loughborough University, Leicestershire, LE11 3TU, UK (Email: \{j.zhang3, g.zheng\}@lboro.ac.uk).}
  \thanks{I. Krikidis is with the    Department of Electrical and Computer Engineering,  University of Cyprus, 1678 Nicosia Cyprus. (E-mail: krikidis@ucy.ac.cy).}
 \thanks{R. Zhang is with the Department of Electrical and Computer Engineering, National University of Singapore. (Email:elezhang@nus.edu.sg).}
}

 \maketitle

\begin{abstract}
This paper investigates the optimal transmit beamforming design of simultaneous wireless information and power transfer (SWIPT)  in the multiuser multiple-input-single-output (MISO) downlink with specific absorption rate (SAR) constraints. We consider the power splitting  technique for SWIPT, where each receiver divides the received signal into two parts: one for information decoding and the other for energy harvesting with a practical non-linear  rectification model. The problem of interest is to maximize as much as possible the received signal-to-interference-plus-noise ratio (SINR) and the energy harvested for all receivers, while satisfying the   transmit power and the SAR constraints by optimizing the transmit beamforming at the transmitter and the power splitting ratios at different receivers. The optimal beamforming and power splitting solutions are obtained with the aid of  semidefinite programming and bisection search. Low-complexity fixed beamforming and hybrid beamforming techniques are also studied. Furthermore, we study the effect of imperfect channel information and radiation matrices, and design robust beamforming to guarantee the worst-case  performance. Simulation results demonstrate that our proposed algorithms can effectively deal with the radio exposure constraints and significantly outperform  the conventional transmission scheme with power backoff.
\end{abstract}


\section{Introduction}

Simultaneous wireless information and power transfer (SWIPT) is a new technology where information and energy   flows co-exist, co-engineered to simultaneously provide communication connectivity and energy sustainability \cite{Ioannis-ComMag-14, CLE}.  It has been considered as a new promising solution to transmit information and energy to low power devices and to extend the battery lifetime of wireless networks, especially in wireless sensor networks and Internet of Things (IoT) applications. Compared to the traditional energy harvesting (EH) and green communication techniques, which collect energy from natural and man-made sources such as solar, wind or mechanical vibration, SWIPT can be fully controlled and optimized by harvesting energy from the radio-frequency (RF) signals. From the seminal work of Varshney \cite{VAR}, who introduced the concept of SWIPT and the fundamental trade-off between information and energy transfer (i.e., information-energy capacity region), substantial  works appear in the literature that study SWIPT from  different perspectives.

Initial works in SWIPT assume a linear channel to transfer both information and power and investigate sophisticated transmission techniques and/or receiver architectures for SWIPT.  The work in \cite{PER} characterizes the fundamental trade-off between information and energy for a basic multiple-access channel with information and EH receivers; it has been shown that a feedback channel is a critical mechanism to increase the information-energy capacity region.  To enable practical   techniques to convey information and power, the authors in \cite{ZHA} propose two practical receiver approaches, namely, a) ``time switching (TS)'', where the receiver switches between decoding information and harvesting energy,  and b) ``power splitting (PS)'', where the receiver splits the received signal into two parts for decoding information and harvesting energy, respectively. Using the PS approach and multiple transmit antennas at the transmitter, the optimal multiple-input single-output (MISO) beamforming with both quality-of-service (QoS) and EH constraints was studied in \cite{Stelios-TWC-14} and  \cite{Zhang-14} for an interference channel and a downlink MISO channel, respectively. To reduce the complexity of the SWIPT receivers, some studies overcome the use of RF-to-DC circuits  and achieve information transfer by embedding information into the shape of the transmitted signal rather than modulating radio waveforms. The work in \cite{RUI1} introduces the integrated receiver, where information is embedded in the amplitude of energy signals,  while decoding is performed by taking samples at the output of the rectification circuit. The authors in \cite{ZHA2} also exploit the concept of pre-coded spatial modulation and convey information both at the energy pattern of the transmitted signal and the index of the received antenna. More recent works take into account the non-linearity of the rectification process and focus on the waveform design for SWIPT. By introducing a simple and tractable model of the diode nonlinearity (i.e., physics-based diode model), the work in \cite{CLE3} deals with the design of waveforms that maximizes the DC power at the output of the rectifier. This work is extended in \cite{CLE4} to convey both information and energy by transmitting a superposition of unmodulated and modulated multi-tone signals.

On the other hand, wireless technologies are characterized by terminals that are subject to strict regulations on the level of RF radiation that users of the terminals are exposed to. Since RF radiation has been proven harmful for humans and the environment, these regulations minimize the potential biological effects (e.g., tissue heating, metabolic changes in the brain, carcinogenic effects) caused by RF radiation \cite{VOL,IARC}. Two widely adopted regulations/measures on RF exposure are mainly considered i.e., {{the first one is the maximum permissible exposure (MPE) in [W/$m^2$], which is defined as the highest level of electromagnetic radiation (EMR)   to which a user may be exposed without incurring an established adverse health effect.
 This is a particularly important issue in the applications of wireless
power transfer (WPT) and various relevant studies have been carried out concerning EMR. For instance, the problem of scheduling the power  chargers  is investigated in \cite{safe_charging} so that the charging utility for all rechargeable devices
is maximized with a constraint on  EMR.
The works in \cite{safe_distributed} and \cite{scheduling_charging} deal with the problem of maximizing the harvested energy and wireless charging tasks scheduling, respectively, when the transmitted signals guarantee a well-defined EMR constraint.
The work \cite{secure_charging} discusses the security issues in RF-based WPT systems and investigates techniques that ensure confidentiality, security and safety for different types of attacks.}}
The second one is  the specific absorption rate (SAR)  that  measures the absorbed power in a unit mass of human tissue by using units of Watt per kilogram [W/kg].  The Federal Communications Commission (FCC) enforces an SAR limitation of $1.6$ W/kg averaged over one  gram of tissue on partial body exposure \cite{FCC-01};  other regulatory agencies (e.g., Comit\'e Europ\'een de Normalisation \'Electrotechnique in European Union) also adopt a similar limitation of $2$ W/kg averaged over 10 grams of tissue on SAR measurements \cite{ULC}. For short distances (e.g., the distance between the transmitter and the nearest user is less than $20$ cm), the SAR measure dominates the RF exposure and {{therefore becomes a more critical factor than the MPE}} for the design of efficient uplink communications. {{SAR is obviously different from MPE in that it is a different measure of RF exposure and applies to short distances. In addition, it needs to account for various exposure constraints on the whole body, partial body, hands, wrists, feet, ankles, etc., with various measurement limitations according to FCC regulations \cite{FCC-01}, therefore multiple SAR constraints are
needed even for a single transmit device. For instance, the iPhone 6 Model A1549 has a whole body SAR of 1.14 W/kg and head SAR of 1.08 W/kg \cite{iphone}.}}

The integration of RF exposure constraints, and specifically the SAR constraint, in the design of wireless communication systems is a timely research area with significant impact; however, few works in the literature take  into account SAR regulations.  In conventional single-antenna wireless communication systems, the SAR exposure limitation simply introduces an additional transmit power constraint, and can be easily guaranteed by reducing the transmit power below a specific threshold. However, it becomes more involved when there are multiple transmit antennas and signals can be beamformed towards certain directions. In this case, the exposure limitation is coupled with the multi-antenna channel vector/matrix, so it needs to be incorporated into the transmitter optimization.
{{The measurements and simulations carried out in \cite{Murch-04} demonstrate SAR as a function of
the phase different between two transmit antennas. The SAR reduction and modeling in multi-antenna systems has been studied in \cite{Qiang-07,Mahmoud-08}.} The SAR constraints are incorporated into the transmit signal design for a multiple-input multiple-output (MIMO) uplink channel in \cite{Hochwald-12}, where the quadratic model for the SAR measurements is first proposed.
In \cite{Hochwald-14}, a SAR code is proposed to handle a  SAR-constrained channel, which is shown to have significant improvement over the conventional Alamouti space-time code
under SAR limitations. It is also revealed that the SAR measurement is a function of the quadratic form of the transmitted signal with the SAR matrix. The SAR-aware beamforming and transmit signal covariance optimization methods are first presented in \cite{Ying-CISS-13} and \cite{Ying-Globecom-13}. Capacity analysis with multiple SAR constraints on single-user MIMO systems is intensively examined in \cite{Ying-TWC-13}.   Sum-rate analysis for a multi-user MIMO system with SAR constraints is performed in \cite{Ying-TWC-17} with both perfect and statistical channel state information (CSI). Form an information theoretical perspective,  SAR-constrained multi-antenna transmit covariance optimization can be considered as the classical MIMO channel capacity optimization problem subject to generalized linear transmit covariance constraints  studied in \cite{Rui-TIT-12}.

The integration of SAR exposure constraints in the design of wireless power transfer or  SWIPT systems is an unexplored research area \cite{Rui-TCOM-17}. Although SWIPT corresponds to a controlled transmission of RF radiation to communicate and energize, and may significantly contribute to the electromagnetic pollution (electrosmog), no work in the literature discusses the integration of SAR in SWIPT.  For example,  one of the main application areas of SWIPT is for medical devices in  wireless body area networks, where an access point will support the communication connectivity and the power sustainability of a short-range sensor network in, on, or around the human body \cite{WANG, XZHA, ASIF}. Since the access-point may be placed within $20$ cm of the body, SAR measurements are mandatory and should be taken into account in the SWIPT design. To the best of our knowledge, the impact of the SAR on SWIPT has not been studied before.

In this paper, we focus on the PS approach and   study the beamforming design in  a multiuser MISO downlink channel, where the receivers are characterized by both communications QoS and EH constraints with additional SAR limitations; the nonlinearity of the rectification circuit is also taken into account.
The contributions of this paper are summarized as follows:
\begin{enumerate}
    \item We introduce and formulate the SAR-aware beamforming optimization problem with simultaneous  QoS and  EH requirements.
	\item We derive the optimal beamforming solution in the general case by leveraging semidefinite programming (SDP) together with rank relaxation. More importantly, we prove that the proposed approach always gives rank-$1$ solutions, such that the exactly optimal beamforming solutions can be derived. In the special case of a single receiver with a single SAR constraint, we propose a fast algorithm to achieve the optimal solution.
    \item We develop low-complexity suboptimal solutions using both fixed beamforming schemes based on the criterion of  zero-forcing (ZF) or   regularized ZF (RZF), and a hybrid beamforming scheme  that combines the ZF and maximum ratio transmission (MRT)   to achieve a better performance-complexity tradeoff.	
	\item Furthermore, when the CSI and/or the SAR matrices are imperfectly known at the transmitter, we devise a robust beamforming scheme suitable for imperfect CSI and SAR matrices with bounded errors, to guarantee the QoS and EH requirements, as well as the SAR constraints in the worse case.
\end{enumerate}

In addition, we have made the following assumptions:
\begin{itemize}
    \item The input signal is circularly symmetric complex Gaussian (CSCG) distributed. This is in general sub-optimal and the input distribution under nonlinearity can be optimized as well, as studied in  \cite{CLE3,learning_signal,VAR}.
    \item We adopt a simplified nonlinear EH model that depends only on the power of the received signal. Note that a more accurate EH model  is a function of the received signal other than just its power, as shown and experimentally demonstrated in the literature \cite{CLE,experiment}.
\end{itemize}

{\underline{\it Notation}:} All boldface letters indicate vectors (lower case) or matrices (upper case). The superscripts $(\cdot)^{T}$, $(\cdot)^\dag$,  $(\cdot)^{-1}$, $(\cdot)^{\ddag}$, $\mbox{diag}(\cdot)$ denote the transpose, the conjugate transpose, the matrix inverse,  the Moore-Penrose pseudo-inverse and the diagonal elements respectively. A circularly symmetric complex Gaussian random variable $z$ with mean $\mu$ and variance $\sigma^2$ is represented as $z\sim\mathcal{CN}\left(\mu, \sigma^2\right)$. The identity matrix of size $M$, and the zero matrix of size $m\times n$, are denoted by $\mathbf{I}_{M}$ and $\mathbf{0}_{m\times n}$, respectively. $\|\mathbf{z}\|$ and $\|\mathbf{z}\|_F$ denote  the $L_2$ and Fibonacci norms of a complex vector $\mathbf{z}$, $|z|$ denotes the magnitude of a complex variable $z$, $\tr{(\mathbf{A})}$ denotes the trace of a matrix $\mathbf{A}$, while $\mathbf{A}\succeq 0$ indicates that matrix $\mathbf{A}$ is positive semidefinite. $i\triangleq \sqrt{-1}$ denotes the imaginary unit.

This paper is organized as follows: Section II sets up the system model and introduces the optimization problem considered. In Section III, we investigate the joint optimization of the optimal beamforming design and power splitting ratios.  Section IV discusses solutions for the fixed and hybrid beamforming schemes, while Section V develops the robust solution  under imperfect estimation of the channels and the SAR matrices.  Simulation results are presented in Section VI, followed by   conclusions in Section VII.

\section{System Model and Problem Formulation}

We assume a MISO  downlink channel consisting of an $N_t$-antenna transmitter (e.g., a base station (BS)) and $K$ single-antenna receivers that employ single-user detection as illustrated in Fig. \ref{fig:sys}. This ensures simple receiver processing due to size limitation or limited processing power.      We assume that the BS transmits with a total power $P_T$ and {{let $s_k$ be its transmitted data symbol to receiver $k$, which is Gaussian distributed with zero mean and unit variance, i.e., $\mathbb{E}\{\|s_k\|^2 \}=1$}.  The transmitted data symbol $s_k$ is mapped onto the antenna array elements by the beamforming/precoding vector $\mathbf{w}_k \in \mathbb{C}^{N_t\times 1}$. 
 {{Other properties of the signal  (e.g., modulation, waveform and input distribution) can also be optimized to improve the efficiency of RF-DC conversion \cite{CLE3,learning_signal}, so the CSCG distribution of the input signal is in general sub-optimal. However,    the joint design of signal and beamforming is a  challenging problem and out of the scope of this paper.}
\begin{figure}
  \centering
  \includegraphics[scale=0.35]{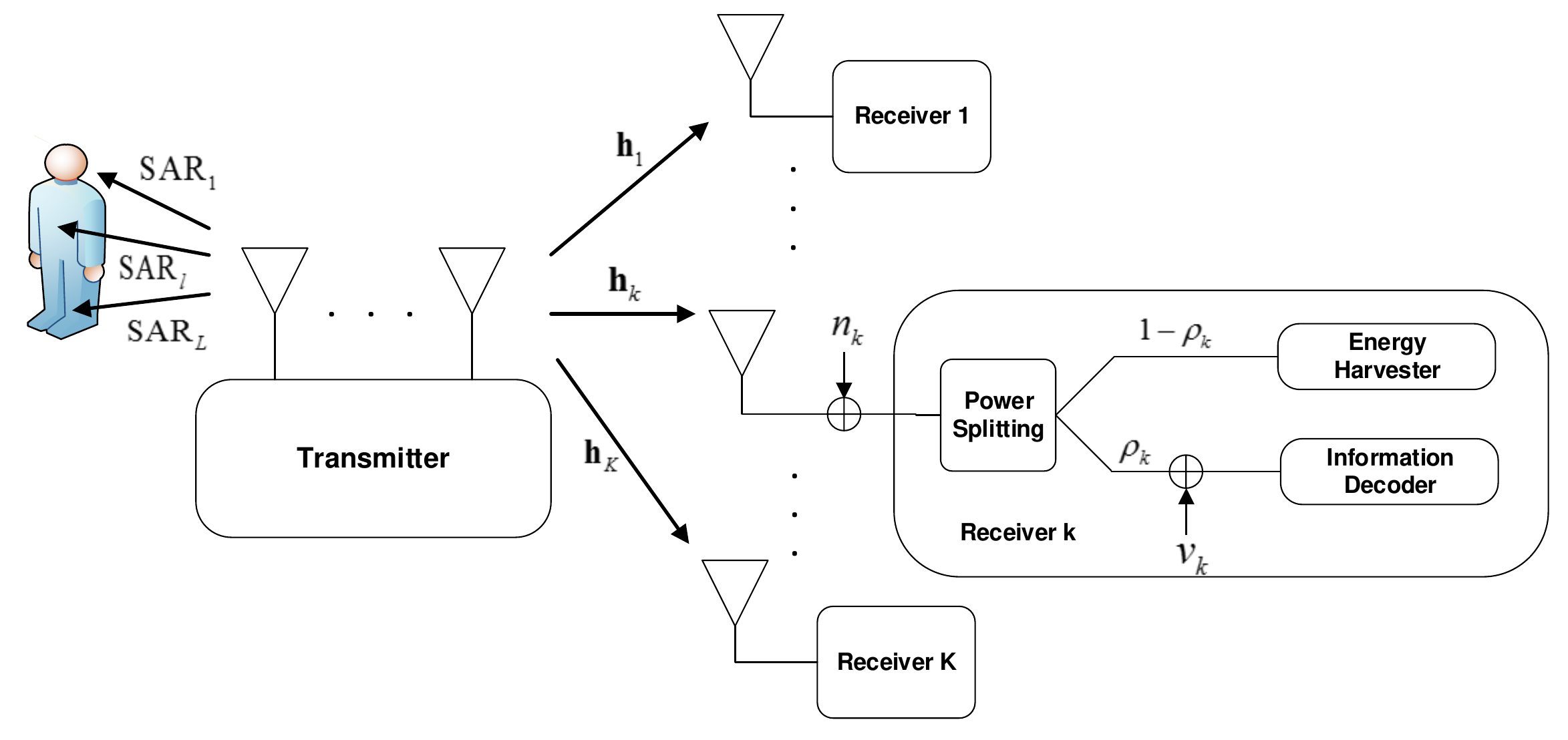}
  \caption{System model of SWIPT with SAR constraints.}\label{fig:sys}
\end{figure}
All wireless links exhibit independent fading and additive white gaussian noise (AWGN) with zero mean and certain variances. The fading is assumed to be frequency non-selective block fading.  This means that the fading coefficients in the vector channel $\mathbf{h}_{k} \in \mathbb{C}^{N_t\times 1}$  remain constant during one slot, but may change independently from one slot to another. 
The mean and variance of the channel coefficients capture  large-scale degradation effects such as path-loss and shadowing. The received     baseband signal at receiver $k$ can be expressed as
\begin{align}\label{sys1}
y_k=\underbrace{\mathbf{h}_{k}^\dag \mathbf{w}_k s_k}_{\textrm{Information signal}}+\underbrace{\sum_{j\neq k}\mathbf{h}_k^\dag \mathbf{w}_js_j}_{\textrm{Interference}}+n_k,
\end{align}
where $n_k$ denotes the AWGN component with zero mean and variance $N_0$. Therefore, the received power at receiver $k$ is equal to
\begin{equation}\label{eq:TotalPower}
P^r_k=\displaystyle \sum_{j=1}^{K} |\mathbf{h}_{k}^\dag \mathbf{w}_j|^2   + N_0.
\end{equation}

The receivers have RF-EH capabilities and therefore can harvest energy from the received RF signal based on the power splitting technique \cite{ZHA}, {{which is a mature SWIPT technique and does not require  strict time synchronization between information and power transfer\footnote{{ In case of nonlinear energy harvesting, neither the power splitting technique nor the time splitting technique always performs better.  In the lower power regime, the time switching technique performs better than the power splitting technique \cite{tsps}.}}}}. With this approach, each receiver splits its received signal into two parts: a) one part is converted to a baseband signal for further signal processing and data detection, and b) the other part is driven to the required circuits for conversion to DC voltage and energy storage. Let $\rho_k \in (0,1)$ denote the power splitting parameter for the $k$-th receiver; this means that $100 \rho_k \%$ of the received power is used for data detection, while the remaining amount is the input to the RF-EH circuitry.  More specifically, after reception of the RF signal at the receiver, a power splitter divides the power $P^r_k$ into two parts according to $\rho_k$, so that $\rho_k P^r_k$ is directed towards the decoding unit and $(1 - \rho_k)P^r_k$ towards the EH unit. During the baseband conversion, additional circuit noise, $v_k$, is present due to phase offsets and non-linearities which is modeled as AWGN with zero mean and variance $N_C$.

The signal-to-interference-plus-noise ratio (SINR) metric characterizing the data detection process at the $k$-th receiver is given by
\begin{equation}
\Gamma_k= \frac{\rho_k  |\mathbf{h}_k^\dag \mathbf{w}_k|^2}{\rho_k\left(N_0+\sum_{j\neq k} |\mathbf{h}_{k}^\dag\mathbf{w}_j|^2\right) + N_C}.
\label{eq:Gamma}
\end{equation}
On the other hand, the total   power that can be harvested is equal to $P_k^S = F( (1-\rho_k) P^r_k)$,  where $F(\cdot)$ is a non-linear parametric EH   function.  Note that in general $F(\cdot)$ should be a function $y_k$ rather than just the power of the energy signal. In this paper we adopt a simplified model to highlight the dependency on the power of the energy signal only.  More details about $F(\cdot)$  will be discussed later.

As discussed in the introduction,  wireless communication devices are  subject to SAR limitations.
Previous reported results such as \cite{Hochwald-14} have    shown that the pointwise
SAR value with multiple transmit antennas can be modeled as a quadratic form of the transmitted signal,  and the SAR matrix fully describes the SAR measurement’s dependence on the transmitted signals;  entries of the SAR matrix have units of W/kg. Because the SAR measurements are always real positive numbers, the SAR matrices  are positive-definite conjugate-symmetric  matrices.

Since SAR is a quantity averaged over the transmit signals, we model the $l$-th SAR constraint with a time-averaged quadratic constraint  given by
\be
\mbox{SAR}_l = \mathbb{E}_{\{s_k\}}\tr\left(\sum_{k=1}^K s_k^\dag \qw_k^\dag \qA_l \mathbf{w}_k s_k\right) = \sum_{k=1}^K \qw_k^\dag \qA_{l} \qw_k\le P_l,
\ee
where $\qA_l\succeq \qzero$ is the $l$-th SAR matrix and $P_l$ is the $l$-th SAR limit.

\subsection{Problem Formulation}

 We study the following problem of maximizing the ratios of the received SINR and EH over the target requirements, i.e., $\max\{\frac{\Gamma_1}{\bar\gamma_1}, \cdots, \frac{\Gamma_k}{\bar\gamma_k},\cdots, \frac{\Gamma_K}{\bar\gamma_K}, \frac{P_1^S}{\bar\lambda_1}, \cdots, \frac{P_k^S}{\bar\lambda_k}, \cdots, \frac{P_K^S}{\bar\lambda_K}\}$ subject to the  SAR and total power constraints, where $\bar\gamma_k, \bar\lambda_k$ are the SINR and the EH requirements, respectively. {{The choice of the objective function will balance  the received SINR and the EH between users.}}  To make the problem more tractable, we introduce an auxiliary  variable $t$, and formulate the optimization problem as follows
\bea\label{eqn:prob:01}
\textbf{P1:} && \max_{\{\qw_k, \rho_k,t\}}~t \\
&& \mbox{s.t.}~
   \frac{ |\qh_{k}^\dag\qw_k|^2}{\sum\limits_{j=1,j \ne
k}^K  |\qh_{k}^\dag\qw_j|^2 + N_0 + \frac{N_C}{\rho_k}} \ge t \bar \gamma_k\triangleq \gamma_k,\label{eq:SINR0}\\
&& F\left((1-\rho_k) \left(\sum\limits_{j=1}^K  |\qh_{k}^\dag\qw_j|^2 + N_0  \right)\right) \ge t \bar\lambda_k \triangleq\tilde\lambda_k, \label{eq:EH0}\\
&& 0\le \rho_k\le 1, \forall k, \notag \\
&&   \sum_{k=1}^K \qw_k^\dag \qA_{l} \qw_k \leq   P_{l} , \forall l,\label{SAR:constraint}\\ 
&& \sum_{k=1}^K \|\qw_k\|^2 \le P_T,\notag
\eea
where $P_T$ is the maximum total transmit power, and $L$ is the number of SAR constraints.

 In this paper, we adopt a non-linear parametric EH model, so the output DC power at the $k$-th receiver can be represented by the nonlinear function $F(x)$, where $x$ is the input RF power. {{The nonlinear function can take many forms to capture the relationship between
the input and output power at the energy receiver, such as the sigmoid function \cite{Schober-15}, \cite{Schober-17}, the linear fraction \cite{Yunfei-17}, the piece-wise linear function \cite{Dong-16}, the second order polynomial model \cite{Xu-17}, and a heuristic expression by curve fitting in \cite{Yunfei-16}.}}
 {{We use $F$ to denote  a general nonlinear  energy harvesting function; our analysis is general and independent of the specific function $F$.}} In general, the nonlinear EH function is monotonically increasing, therefore we can find the inverse mapping $F^{-1}(\cdot)$, and the EH constraint \eqref{SAR:constraint} can be rewritten as
\be
    (1-\rho_k)\left(\sum\limits_{j=1}^K  |\qh_{i}^\dag\qw_k|^2 + N_0  \right)\ge F^{-1} (  \tilde\lambda_k)\triangleq \lambda_k.
\ee

In general, it is difficult to solve the above problem \textbf{P1}, because both SINR and EH constraints \eqref{eq:SINR0} and \eqref{eq:EH0} are nonconvex, and we also have  additional multiple SAR constraints. Instead, we solve the following power minimization problem  \textbf{P2} below
\bea\label{eqn:prob:0}
\textbf{P2:} \min_{\{\qw_k, \rho_k\}}&& \sum_{k=1}^K \|\qw_k\|^2  \\
\mbox{s.t.}&&
  \frac{ |\qh_{k}^\dag\qw_k|^2}{\sum\limits_{j=1,j \ne
k}^K  |\qh_{k}^\dag\qw_j|^2 + N_0 + \frac{N_C}{\rho_k}} \ge   \gamma_k,\notag\\
&& (1-\rho_k)\eta_k \left(\sum\limits_{j=1}^K  |\qh_{j}^\dag\qw_k|^2 + N_0  \right)\ge  \lambda_k, \notag\\
&& 0\le \rho_k\le 1, \forall k, \notag \\
&&   \sum_{k=1}^K \qw_k^\dag \qA_{l} \qw_k \leq   P_{l}, \forall l.\notag  
\eea

Clearly, the problem \textbf{P2}  is closely related to the problem \textbf{P1}. For instance, if \textbf{P1} is solved and the optimal $t^*$ is achieved, then the same beamforming and power splitting solutions are also optimal for \textbf{P2} to achieve the same SINR and EH, and the optimal minimum power will be $P_T$. Because \textbf{P1} is a quasiconvex problem in $t$, once \textbf{P2} is solved, \textbf{P1} can be solved via a bisection search algorithm over $t$ as follows (note that $t$ in \textbf{P1} is embedded in $\gamma_k, \lambda_k$ of \textbf{P2}):

 \textbf{\underline{Algorithm 1}} to Solve \textbf{P1}:
 \begin{enumerate}
    \item Set the upper and lower bounds of $t$ as $t^U$ and $t^L$. Repeat the following steps until convergence.

    \item Calculate $t= \frac{t^U+t^L}{2}$, and solve \textbf{P2} with $\gamma_k=t \bar \gamma_k$ and $ \lambda_k=F^{-1} (t \bar\lambda_k)$.

    \item If \textbf{P2} is feasible and $\sum_{k=1}^K \|\qw_k\|^2\le P_T$, $t^L = t$; otherwise $t^U=t$.

    \end{enumerate}

Therefore, in the rest of the paper we will focus on solving the problem \textbf{P2}.

\section{The Optimal   Scheme}
In this section, we investigate the optimal beamforming and power splitting solutions to \textbf{P2}, and we begin with developing a fast solution for the special case of \textbf{P2} referring to a single-receiver system with a single SAR constraint, i.e., $K=L=1$.

\subsection{Fast Solution to the Special Case of $K=L=1$}
When we have only   a single user and a single SAR constraint, the problem \textbf{P1} reduces to the problem below, where the receiver and SAR indices are dropped for the sake of simplicity,
\bea\label{eqn:prob:1}
\textbf{P3:} \min_{\{\qw,\rho\}}&& \|\qw\|^2  \\
\mbox{s.t.}
&&\Gamma=\frac{|\qh^\dag\qw|^2}{N_0 + \frac{N_C}{\rho}} \ge \gamma, \label{eqn:SU:SINR}\\
&&  (1-\rho) \left( |\qh^\dag\qw|^2 + N_0  \right)\ge \lambda, \label{eqn:SU:EH}\\
&&   \qw^\dag \qA \qw \leq   P.  \label{eqn:SU:SAR}
\eea
We find the following Lemma is useful to derive the optimal solution to  \textbf{P1}.
\begin{lemma}
 Both the SINR and the EH constraints in \eqref{eqn:SU:SINR} and \eqref{eqn:SU:EH} must be satisfied with equality at the optimum of \textbf{P3}.
\end{lemma}
\begin{proof}
It can be proved by contradiction.  If \eqref{eqn:SU:SINR}  is not tight, $\rho$ can be reduced without affecting \eqref{eqn:SU:EH}; if \eqref{eqn:SU:EH} is not tight, $\rho$ can be increased without affecting \eqref{eqn:SU:SINR}. This completes the proof.
\end{proof}

Base on Lemma 1, we know that the optimal solution to \textbf{P1} falls into the following two cases:
\begin{itemize}
  \item \textbf{Case I}: Only the SINR and the EH constraints are   satisfied with equality.

  In this case, we first solve the optimal $\bar\rho$. From the equality constraint \eqref{eqn:SU:SINR}, we can get
  \be
    |\qh^\dag\qw|^2 = \gamma(N_0 + \frac{N_C}{\rho}).\label{eq:hw}
  \ee
  Substituting \eqref{eq:hw} into \eqref{eqn:SU:EH}, we can have the following equation, i.e.,
  \be
    (1-\rho) (\gamma N_0 + \gamma\frac{N_C}{\rho} + N_0) = \lambda,
  \ee
  which leads to a quadratic equation of $\rho$ given by
\bea\label{eqn:rho}
     -(\gamma+1) N_0 \rho^2 + ((\gamma+1)N_0-\gamma N_C - \lambda)\rho + \gamma N_C=0.
\eea
Its solution is written as
\be
    \bar\rho = \frac{((\gamma+1)N_0-\gamma N_C - \lambda) + \Delta}{2(\gamma+1) N_0},\label{eq:rho}
\ee
where $\Delta\triangleq \sqrt{ ((\gamma+1)N_0-\gamma N_C - \lambda)^2 +4\gamma(1+\gamma) N_0 N_C  }$.
It is easy to verify that $\bar\rho$ is between 0 and 1, because the left hand side of the quadratic equation \eqref{eqn:rho} is strictly negative when $\rho=1$ and strictly positive when $\rho=0$. Note that $\bar\rho$ is the optimal solution regardless of $\qw$ as long as \textbf{P3} is feasible.

Once $\bar\rho$ is found, the optimal $\qw$ can be solved by satisfying the equality \eqref{eq:hw} and minimizing the objective simultaneously, which is given by
\be
    \bar\qw = \frac{\sqrt{\gamma(N_0 + \frac{N_C}{\bar\rho})}\qh}{\|\qh\|^2}.
\ee

 After solving $\bar\qw$, we can substitute it into \eqref{eqn:SU:SAR} and check whether the SAR constraint   is satisfied: if so, $\bar\qw$ is the optimal beamforming solution; otherwise, we need to consider the next case.

  \item \textbf{Case II}: All three SINR, EH and SAR constrains  are  satisfied with equality.

 In this case, we can simplify \textbf{P3} to the following problem where  $\bar\rho$ is given by \eqref{eq:rho}
\bea\label{eqn:prob:11}
\textbf{P4:} \min_{\{\qw,\rho\}}&& \|\qw\|^2 \label{eq:SU:P4} \\
\mbox{s.t.}
&& |\qh^\dag\qw|^2  \ge \gamma (N_0 + \frac{N_C}{\bar\rho})\triangleq \bar c, \notag\\
&&   \qw^\dag \qA \qw \leq   P. \label{eq:waw}  
\eea

We keep the inequality constraints for the convenience of the derivation. In general, it is difficult to find the optimal solution to \textbf{P4}  directly. We first study the analytical beamforming structure that   facilitates the development of an efficient solution.

The Lagrangian of \textbf{P4} is given by
\bea
L &  = &\|\qw\|^2  +  a\left( \gamma \left(N_0 + \frac{N_C}{\bar\rho}\right) - |\qh^\dag\qw|^2 \right)\notag\\
&& +\qw^\dag(\qI + b\qA - a \qh\qh^\dag)\qw +  a  \gamma (N_0 + \frac{N_C}{\bar\rho}) - b P,\notag\\
&& + b(\qw^\dag \qA \qw -  P)
\eea
where $a\ge 0$ and $b\ge 0$ are dual variables.
The dual problem is then formulated as:
\bea
    \max_{a,b\ge0} && a  \gamma (N_0 + \frac{N_C}{\bar\rho}) - b P \\
    \mbox{s.t.} &&  \qI + b\qA - a \qh\qh^\dag \succeq \qzero,
\eea
and the optimal $\qw$ admits the following form
\be
    \qw= a(\qI + b\qA)^{-1}\qh.
\ee

Let the eigenvalue decomposition of $\qA$ be $\qA = \qU \qD \qU^\dag$, where $\qU$ is a unitary matrix and $\qD$ is a diagonal matrix with elements being the eigenvalues of $\qA$.  Then we have $(\qI + b\qA)^{-1} = \qU(\qI + b \qD)\qU^\dag$, and
\be
    \qw = a\qU(\qI + b \qD)^{-1}\qU^\dag\qh= a\qU(\qI + b \qD)^{-1}\tilde \qh,~~ \mbox{with}~ \tilde \qh \triangleq \qU^\dag\qh,\label{eq:w2}
\ee
and
\be
 |\qh^\dag \qw | = \|\qh^\dag\qU(\qI + b \qD)^{-1}\tilde \qh\| = a \tilde \qh^\dag(\qI + b \qD)^{-1}\tilde \qh = \sqrt{\bar c},\label{eq:hw2}
\ee
where $\bar c$ is defined in \textbf{P4}.
From \eqref{eq:hw2}, we can solve
\be
a^2 = \frac{\bar c }{(\tilde \qh^\dag(\qI + b \qD)^{-1}\tilde \qh )^2}.\label{eq:a}
\ee

Now we proceed to consider the constraint \eqref{eq:waw}  and define the function below:
\bea
    f(b)& = & \qw^\dag \qA \qw - P \notag\\
       &\overset{(a)}{=}  &a^2 \tilde \qh^\dag (\qI + b \qD)^{-1} \qU^\dag\qA \qU(\qI + b \qD)^{-1}\tilde \qh- P\notag\\
    &\overset{(b)}{=}&a^2  \tilde \qh^\dag (\qI + b \qD)^{-1}\qD(\qI + b \qD)^{-1}\tilde \qh- P\notag\\
      &\overset{(c)}{=}& \frac{\bar c \tilde \qh^\dag (\qI + b \qD)^{-1}\qD(\qI + b \qD)^{-1}\tilde \qh}{(\tilde \qh^\dag(\qI + b \qD)^{-1}\tilde \qh )^2}- P,\label{eq:fb}
\eea
where $(a)$ is due to \eqref{eq:w2},  $\qA = \qU \qD \qU^\dag$ is used to obtain $(b)$, and we use \eqref{eq:a} to get $(c)$.

Clearly, the parameter $b$ that satisfies $f(b)=0$ uniquely determines the optimal solution. We find the following theorem useful to develop an efficient numerical solution to solve $f(b)=0$.

\begin{theorem}
  $f(b)$ in \eqref{eq:fb} is a decreasing function in $b$.
\end{theorem}

 \begin{proof}
Suppose $\qD =\mbox{Diag}(d_1, \cdots, d_{N_t})$ and $|\tilde\qh|=[h_1, \cdots, h_{N_t}]$. Then after some algebraic  manipulation,   $f(b)$ becomes
\bea
     &&f(b) = \frac{  ^\dag (\qI + b \qD)^{-1}\qD(\qI + b \qD)^{-1}\tilde \qh}{ \tilde \qh^\dag(\qI + b \qD)^{-1}\tilde \qh  \tilde \qh^\dag(\qI + b \qD)^{-1}\tilde \qh   }-P\notag\\
    &&=\frac{\sum_{n=1}^{N_t}\frac{h_n^2d_n}{(1+bd_n)^2} }{ \left(\sum_{n=1}^{N_t} \frac{h^2}{1+b d_n}\right)^2}-P\label{eq:fbdecreasing}\\
    &&=\frac{1}{\sum_{n=1}^{N_t}\frac{h_n^2}{d_n}}\left(1+\sum_{m\ne n} \frac{ h_m^2 h_n^2(d_m-d_n)^2 }{(1+b d_m^2)(1+b d_n^2)}\right)-P.\notag
\eea

From \eqref{eq:fbdecreasing}, it can be seen that $f(b)$ is a decreasing function in $b$, and this completes the proof.
\end{proof}

{{Theorem 1 can be verified by the simulation result in Fig. \ref{fig:fb}.}}

\begin{figure}
  \centering
 \includegraphics[width=0.35\textwidth]{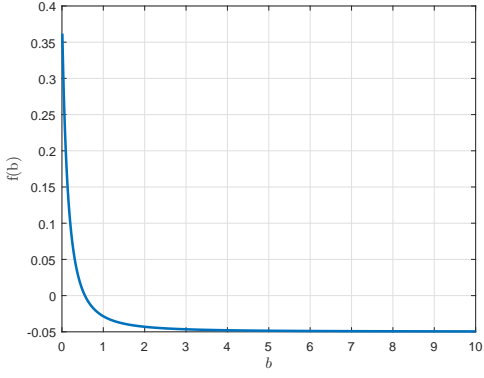}
  \caption{The monotonicity of the function $f(b)$. The parameters used are: ~~$\bar c=0.3685, P=1 W, \tilde \qh=[  -1.6475 + 0.3194i  ~~ -1.0247 - 0.0921i  ~~  0.2358 + 0.1299i ~~   0.2767 - 0.3367i]^T, \qD=\mbox{diag}([    7.4469  ~~   1.8896  ~~   6.8678  ~~   1.8351])$.}\label{fig:fb}
\end{figure}

Based on Theorem 1, we propose the following bisection search algorithm to solve $b$.

 \textbf{\underline{ Algorithm 2}} to Solve $f(b)=0$.

\begin{enumerate}
  \item Initialize $b_{\min}\ge 0$ and $b_{\max}>0$. Repeat the following steps until convergence.
  \item Set $b=\frac{b_{\min}+ b_{\max}}{2}$.
  \item If $f(b)>0$, $b_{\min}=b$; otherwise $b_{\max}=b$.
\end{enumerate}

Once the optimal $b$ is found, the optimal  optimal beamforming solution and the optimal $a$ can therefore be solved via   \eqref{eq:w2} and  \eqref{eq:a}, respectively.
\end{itemize}

\subsection{The Optimal Solution using SDP in the General Case}

 The general case of the problem \textbf{P2} is nonconvex and difficult to solve. In this section, we develop an efficient SDP based algorithm that jointly optimizes the beamforming vectors and power splitting parameters. To tackle \textbf{P2}, we first introduce new matrix variables $\qW_k=\qw_k\qw_k^\dag$, and then  \textbf{P2} can be reformulated as
\bea\label{eqn:prob:1}
\textbf{P5:} \min_{\{\qW_k, \rho_k\}}&& \sum_{k=1}^K  \tr(\qW_k) \\
\mbox{s.t.}&&
    \frac{ \tr(\qh_{k}\qh_{k}^\dag\qW_k) }{\sum\limits_{j=1}^K  \tr(\qh_{k}\qh_{k}^\dag\qW_j) + N_0 + \frac{N_k}{\rho_k}} \ge \frac{\gamma_k}{1+\gamma_k},\notag\\
&&   \sum\limits_{j=1}^K  \tr(\qh_{k}\qh_{k}^\dag\qW_j) + N_0   \ge \frac{\lambda_k }{1-\rho_k}, \notag\\
&& 0\le \rho_k\le 1, \qW_k\succeq \qzero, \forall k, \notag \\
&&  \sum_{k=1}^K \tr( \qA_l \qW_k) \leq   P_{l} , \forall l. \label{eq:P5:SAR}
\eea

The problem \textbf{P5} is convex because it is linear in all $\{\qW_k\}$ , and both terms $\frac{1}{\rho_k}$ and $\frac{1}{1-\rho_k}$ are convex in $\rho_k>0$. It can be efficiently solved using numerical software package such as CVX \cite{cvx}. Once \textbf{P5} is optimally solved, if the resulting solutions $\{\qW_k\}$  are all rank-1, then they are
the exactly optimal solutions; otherwise, the solutions only  provide a lower bound for the minimum required transmit power.

It has been proved in \cite{Zhang-14} that without the SAR constraints \eqref{eq:P5:SAR}, the above SDP with rank relaxation will produce rank-1 solutions $\{\qW_k\}$ which means they are also optimal to the problem \textbf{P2}. However, whether the SDP with rank relaxation can generate the optimal solution highly depends on the problem structure. With the additional SAR constraints, it is unknown whether this property remains true. In the following theorem, we prove that this is still the case.

\begin{theorem}\label{theo2}
    The optimal solution to \textbf{P5} satisfies $\mbox{rank}(\qW_k)=1, \forall k$, i.e., the SDP relaxation is tight, and the optimal solution to the problem \textbf{P2} can be recovered from $\{\qW_k\}$  via eigenvalue decomposition.
\end{theorem}
The proof is given in Appendix A.

\section{Sub-optimal beamforming schemes}

Although the algorithm based on SDP in Section III ensures the optimal beamforming and power splitting solutions, its complexity is high. In this section, we focus on low-complexity suboptimal solutions including both the fixed heuristic beamforming schemes and the hybrid beamforming scheme that considers  a linear combination of  fixed beamforming strategies.

\subsection{Solutions using Fixed Beamforming}

Suppose the fixed beamforming vector is given by $\qw_k=\sqrt{p_k}\qw^f_k, \|\qw^f_k\|=1$, where $p_k$ is the power for the $k$-th receiver. Let $G_{k,j}\triangleq|\mathbf{h}_{k}^{\dag}\mathbf{w}_{j}^{f}|^2$
denote the link gain between the BS and the $k$-th receiver, and $F_{k,l} = {\qw_k^{f}}^\dag \qA_{l} \qw_k^{f}$ denote the radiation channel gain due to the transmission intended for the $k$-th receiver.

Using the fixed beamforming and rearranging the terms in the problem \textbf{P2}, we have
{\small
\begin{subeqnarray} \label{eq:OFW} \slabel{eq:OFW1}
\vspace{-0.5cm}
\textbf{P6}:~~~ \displaystyle{\min_{ \{p_k,\rho_k\}}} && \sum_{k=1}^{K}p_{k}\\
\slabel{eq:OFW2} \textrm{s.t.} (\frac{1}{\gamma_{k}}+1) G_{k,k} p_{k} & \geq & \frac{1}{\rho_k}N_C +  {\displaystyle \sum_{j=1}^{K}}G_{k,j}p_{j} + N_0 , \forall k,\\
\slabel{eq:OFW3}   {\displaystyle \sum_{j=1}^{K}}G_{k,j}p_{j} +N_0  & \geq & \frac{\lambda_{k}}{1-\rho_{k}}, \forall k,\\
\slabel{eq:OFW5} p_{k} \ge 0,~&& 0 \le \rho_{k} \le 1, \forall k,\\
\sum_{k=1}^K p_k  F_{k,l} &\leq &  P_{l}, \forall l.
\end{subeqnarray}}
The optimization problem \textbf{P6} is convex because it is comprised of a linear objective function and convex constraints. Constraint (\ref{eq:OFW2}) is convex because the term $\frac{1}{\rho_k}$ is convex for $\rho_k>0$ and the other terms are linear, (\ref{eq:OFW3}) is a restricted hyperbolic constraint and the term $\frac{1}{1-\rho_k}$ is also convex.  Note that by solving the problem \textbf{P6}, we obtain optimal values for the splitting parameters, as well as an optimal power allocation for any given  fixed beamforming vectors.

Commonly used fixed beamforming vectors include   the  MRT,   ZF and RZF criteria \cite{RZF}, defined respectively below:
\bea
     \qw_k^{MRT} &=&  \frac{\qh_k}{\|\qh_k\|}, \\
 \mathbf{w}_k^{\text{ZF}} &=& \frac{\left(\mathbf{I}_{N_t} - \mathbf{H}_k^\ddag \mathbf{H}_k \right)\mathbf{h}_k}{\left\|\left(\mathbf{I}_{N_t} -\mathbf{H}_k^\ddag \mathbf{H}_k \right)\mathbf{h}_k \right\|},\\
    \qw_k^{RZF} &=&  (K\qI + \qH\qH^\dag + \sum_{l=1}^L   \qA_l)^{-1} \qh_k,
\eea
where we have defined $\qH=[\qh_1,\hdots,\qh_K]$, and $\mathbf{H}_k = [\mathbf{h}_1, \hdots, \mathbf{h}_{k-1}, \mathbf{h}_{k+1}, \hdots, \mathbf{h}_K]$. The MRT beamforming and the ZF beamforming aim to enhance the received signal strength and remove interference, respectively, while the RZF beamforming provides a tradeoff between them. Note that in the RZF beamforming, the SAR matrices are included to incorporate the radiation constraints.

As a special case, for the ZF beamforming, the optimization problem can be further simplified  because $G_{i,j}=0$, for $i \neq j$. Hence the   optimization problem  \textbf{P6}
simplifies into the following formulation
\bea \label{eq:ZF}
 \textbf{P7:}~~~  \displaystyle{\min_{\{p_k, \rho_k\}}} && \sum_{k=1}^{K}p_{k}\\
  \textrm{s.t.} \frac{\rho_{k}G_{k,k} p_{k}}{\rho_k N_0 + N_C} & \geq & \gamma_{k}, \forall k,\notag\\
 (1-\rho_{k})(G_{k,k}p_k+N_0) & \geq & \lambda_{k}, \forall k,\notag\\
  p_{k} \ge 0,&&0 \le \rho_{k} \le 1, \forall k,\notag\\
\sum_{k=1}^K p_k  F_{k,l} &\leq &  P_{l}, \forall l.\notag
\eea
The problem \textbf{P7} is also convex but is much easier to solve than \textbf{P6} because of the simplified constraints. {{Unlike the previous works \cite{Stelios-TWC-14,Zhang-14} where closed-form solutions to \textbf{P7}  were derived, the SAR constraint does not permit a closed-form solution, therefore we use CVX \cite{cvx} to solve both  \textbf{P6}  and \textbf{P7}. }}

\subsection{Hybrid  Beamforming}
In this subsection, we introduce the hybrid beamforming scheme to provide a   tradeoff between enhancing the received SINR and energy harvested, dealing with interference, and accounting for the SAR constraints, which admits the following expression
\be\label{eq:BF:hybrid}
    \qw_k^{\text{hyb}} =  \sqrt{x_k} \qw_k^{ZF}  + \sqrt{y_k} \qw_k^{MRT},
\ee
where $x_k$ and $y_k$ are combining coefficients.

 Then
\begin{equation}
G_{i,j}=|\mathbf{h}_{i}^\dag \mathbf{w}_j^{\text{hyb}}|^{2}=\begin{cases}
y_{j}|\mathbf{h}_{i}^{\dag} \qw_i^{MRT} |^{2}, & i\neq j\\
|\sqrt{x_{i}}\qh_i^\dag\qw_i^{ZF}  + \sqrt{y_{i}}\mathbf{h}_{i}^{\dag}\qw_i^{MRT} |^{2}, & i=j.
\end{cases}\label{eq:chGainsMRTZF1}
\end{equation}

Define $Q_{ij} \triangleq |\mathbf{h}_{i}^\dag\qw_i^{MRT}| $, $q_i\triangleq |\qh_i^\dag\qw_i^{ZF}|=\|\qw_i^{ZF}\|^2$, and $r_i=|\qh_i^\dag\qw_i^{ZF} \qh_{i}^{\dag}\qw_i^{MRT}|$, then
\bea
G_{i,j}&=&|\mathbf{h}_{i}^\dag \mathbf{w}_j^{\text{hyb}}|^{2}\notag\\
&=&\begin{cases}
y_{j}Q_{ij}^2, & i\neq j\\
|\sqrt{x_{i}}q_i + \sqrt{y_{i}}  {Q_{ii}}|^2 = x_i q_i^2 + y_i Q_{ii}^2 + 2 \sqrt{x_i y_i}r_i, & i=j.\notag
\end{cases}\label{eq:chGainsMRTZF1}
\eea

The transmit power intended for the $k$-th receiver is
\bea
    p_k& =& | \sqrt{x_k}\qw_k^{ZF} + \sqrt{y_k} \qw_k^{MRT} |^2\notag\\
    &=&x_k \|\qw_k^{ZF}\|^2 + y_k \|\qw_k^{MRT} \|^2 + 2\sqrt{x_k y_k} \mbox{Re}({\qw_k^{ZF}}^\dag\qw_k^{MRT})\notag\\
    &=& x_k q_k + y_k e_k + 2\sqrt{x_k y_k} f_k,
\eea
where $\|\qw_k^{ZF}\|^2 =q_k$,  and we have defined $e_k\triangleq\|\qw_k^{MRT} \|^2$ and
$f_k = \mbox{Re}({\qw_k^{ZF}}^\dag\qw_k^{MRT})$.

The $l$-th SAR  power can be reformulated as
\bea
    &&\sum_{k=1}^K \qw_k^\dag \qA_{l} \qw_k \notag\\
    &=& \sum_{k=1}^K (\sqrt{x_k} \qw_k^{ZF}  + \sqrt{y_k} \qw_k^{MRT} )^\dag \qA_{l} (\sqrt{x_k} \qw_k^{ZF}  + \sqrt{y_k} \qw_k^{MRT} )\notag\\
    &=&  \sum_{k=1}^K x_k A_{kl} + y_k B_{kl} + 2\sqrt{x_k y_k} C_{kl},
\eea
where we have defined $A_{k,l} \triangleq {\qw_k^{ZF}}^\dag \qA_l\qw_k^{ZF},$ $B_{ij} \triangleq {\qw_k^{MRT}}^\dag \qA_l\qw_k^{MRT},$ and
 $C_{kl} \triangleq   \mbox{Re}({\qw_k^{ZF}}^\dag \qA_l \qw_k^{MRT})$.

Then the power minimization problem \textbf{P2} becomes
\bea
    \textbf{P8:} && \min_{\{x_k,y_k, \rho_k\}} ~ \sum_{k=1}^K  x_k q_k + y_k e_k + 2\sqrt{x_k y_k} f_k\\
    \mbox{s.t.} && \frac{ x_k q_k^2 + y_k Q_{kk}^2 + 2 \sqrt{x_k y_k} r_k }{\sum\limits_{j=1,j \ne
k}^K  y_{j}Q_{kj}^2 + N_0 + \frac{N_C}{\rho_k}} \ge \gamma_k,\notag\\
&& (1-\rho_k) \left(x_k q_k^2   + 2 \sqrt{x_k y_k} r_k+ \sum\limits_{j=1}^K  y_{j}Q_{kj}^2 + N_0  \right)\ge \lambda_k, \notag\\
&& 0\le \rho_k\le 1, \forall k, \notag \\
                && \sum_{k=1}^K  x_k A_{kl} + y_k B_{kl} + 2\sqrt{x_k y_k} C_{kl} \le P_l, \forall l.\notag
\eea

{{The problem \textbf{P8} is not convex because of the nonconvex term $\sqrt{x_k y_k}$ introduced by the EH and the SAR constraints,  which makes the problem more challenging to solve.}} To deal with this difficulty, we introduce $s_k = \sqrt{x_k y_k}$ and relax it to $s_k^2\le x_k y_k$ or $\|[2s_k; x_k-y_k]\|\le x_k+y_k$, then we have the following second-order cone programming (SOCP) problem formulation
\bea
  \textbf{P9:}  &&\min_{x_k,y_k, \rho_k, s_k} ~ \sum_{k=1}^K  x_k q_k + y_k e_k + 2 s_k f_k\\
    \mbox{s.t.} && \frac{ x_k q_k^2 + y_k Q_{kk}^2 + 2 s_k r_k }{\sum\limits_{j=1,j \ne
k}^K  y_{j}Q_{kj}^2 + N_0 + \frac{N_k}{\rho_k}} \ge \gamma_k,\notag\\
&& (1-\rho_k) \left(x_k q_k^2   + 2 s_k r_k+ \sum\limits_{j=1}^K  y_{j}Q_{kj}^2 + N_0  \right)\ge \lambda_k, \notag\\
&& \|[2s_k; x_k-y_k]\|\le x_k+y_k, \forall k, \notag\\
&& 0\le \rho_k\le 1, \forall k, \notag \\
                && \sum_{k=1}^K x_k A_{kl} + y_k B_{kl} + 2s_k C_{kl} \le P_l, \forall l.
\eea
The problem  \textbf{P9} is   convex   and can be optimally solved. We observe that in most cases,  $s_k = \sqrt{x_k y_k}$ holds true which means it is also the optimal solution to \textbf{P8}; otherwise, we will employ  \eqref{eq:BF:hybrid} to form the hybrid beamforming, and then use the  fixed beamforming scheme in Section IV.A to find the optimal power vector and power splitting   by solving  \textbf{P6}.

{{
\textit{\underline{Complexity Analysis}}:   The complexity of the  optimally solution to the general problem \textbf{P5} is dominated by the SDP constraints, and  according to \cite[6.6.3]{Nemirovski}, the complexity of the interior-point algorithm for solving \textbf{P5}  is $O\left(\sqrt{KN_t}\left(K^3N_t^2+K^2N_t^3\right)\right)$.  While the for solving  the fixed beamforming optimization problem \textbf{P6} and the hybrid beamforming problem  \textbf{P9}, their complexities are dominated by the linear and SOCP constraints, and can be expressed as $O\left(\sqrt{3K+L}(3K^3+LK^2)\right)$ \cite[6.6.1]{Nemirovski} and  $O\left(\sqrt{3K+L}(28K^3 + 4LK^2+4K^2+KL)\right)$ \cite[6.6.2]{Nemirovski}, respectively. It can be clearly seen that the suboptimal beamforming schemes   reduces the complexity significantly when compared to the optimal beamforming.
}}

\section{Robust beamforming schemes}
In this section, we study the robust beamforming design, when the channel information and the SAR matrices are imperfectly known due to estimation and measurement errors. Without a robust solution, the SINR constraints and the EH constraints may not be satisfied; in addition, the SAR constraints may be violated.

We model the $k$-th receiver's actual channel as
\be
    \qh_k =\hat \qh_k + \Delta \qh_k, ~~\mbox{with}~ \Delta \qh_k\in \mathcal{H}_k,
\ee
where   $\hat\qh_k$ denotes the CSI estimate  known to the BS.  $\Delta \qh_k$ represents the CSI uncertainty that belongs to the set below
\be
    \mathcal{H}_k \triangleq  \{\Delta \qh_k |    \|\Delta \qh_k\| \le \sigma_k^2\},
\ee
where $\sigma_k^2$ denotes the error bound.

Similarly, the $l$-th SAR matrix is modelled as

\be
    \qA_l =\hat \qA_l + \Delta \qA_l,  ~~\mbox{with}~ \Delta \qA_l\in \mathcal{A}_l,
\ee
where $\hat \qA_l$ is know to the BS. $\Delta \qA_l$ is the SAR uncertainty within the set below
\be
    \mathcal{A}_l \triangleq  \{\Delta \qA_l|   \|\Delta \qA_l \|_F\le \tau_l\},
\ee
 where $\tau_l$ denotes the error bound.

 We assume that the BS has no knowledge about the error channel vectors or the error SAR matrices, except for their error bounds $\sigma^2_k$ and   $\tau_l$. Thus we take  a worst-case approach for the beamforming design to guarantee the resulting
 solution is  robust to all possible channel and SAR uncertainties within the given error sets. The specific robust design problem is to minimize the overall transmit
power $P_T$ for ensuring the receivers' worst-case individual SINR, EH and SAR constraints,  i.e.,
\bea\label{eqn:prob:2}
\textbf{P10:}&& \min_{\{\qW_k, \qw_k, \rho_k\}}~\sum_{k=1}^K  \tr(\qW_k) \\
\mbox{s.t.}&&
    \tr\left(\qh_{k}\qh_{k}^\dag(\qW_k-\gamma_k\sum\limits_{j=1, j\ne k}^K \qW_j) \right)    \ge  \gamma_k\left(N_0 + \frac{N_k}{\rho_k} \right),\notag\\
    &&\forall k,  \Delta \qh_k \in  \mathcal{H}_k,\notag\\
&&   \tr\left(\qh_{k}\qh_{k}^\dag(\sum\limits_{j=1}^K  \qW_j)\right) \ge \frac{\lambda_k}{1-\rho_k}-N_0, \forall k,  \Delta \qh_k \in  \mathcal{H}_k,\notag\\
&&  \sum_{k=1}^K \tr( \qA_l \qW_k) \leq   P_{l} , \forall l, \Delta \qA_l\in \mathcal{A}_l,\notag\\
&& 0\le \rho_k\le 1, \qW_k\succeq \qzero, \forall k, \notag\\
&& \qW_k = \qw_k\qw_k^\dag, \forall k.\label{eq:rank}
\eea

We first ignore the rank constraint \eqref{eq:rank}
 and define $\qQ_k=\qW_k-\gamma_k\sum\limits_{j=1, j\ne k}^K \qW_j$. Then $\textbf{P10}$ becomes
\bea\label{eqn:prob:3}
\textbf{P11:} &&\min_{\{\qW_k, \rho_k\}}~ \sum_{k=1}^K  \tr(\qW_k) \\
\mbox{s.t.}&&
        \hat\qh_{k}^\dag\qQ_k\hat\qh_{k} +         \hat\qh_{k}^\dag\qQ_k\Delta\qh_{k}+        \Delta\qh_{k}^\dag\qQ_k\hat\qh_{k}+   \Delta\qh_{k}^\dag\Delta\qh_{k}
        \ge \notag \\ && \gamma_k\left(N_0 + \frac{N_k}{\rho_k} \right),\forall k,  \Delta \qh_k \in  \mathcal{H}_k,\label{eqn:delta:h:SINR} \\
&&   \hat\qh_{k}^\dag(\sum\limits_{j=1}^K  \qW_j)\hat\qh_{k} +         \hat\qh_{k}^\dag(\sum\limits_{j=1}^K  \qW_j)\Delta\qh_{k}+        \Delta\qh_{k}^\dag(\sum\limits_{j=1}^K  \qW_j)\hat\qh_{k} \notag \\
 &&   +\Delta\qh_{k}^\dag\Delta\qh_{k}         \ge \frac{\lambda_k}{1-\rho_k}-N_0, \forall k,  \Delta \qh_k \in  \mathcal{H}_k,\label{eqn:delta:h:EH} \\
&&  \tr( (\hat \qA_l + \Delta \qA_l) (\sum_{j=1}^K  \qW_k) ) \leq   P_{l} , \forall l, \Delta \qA_l\in \mathcal{A}_l \label{eqn:delta:Radiation}\\
&& 0\le \rho_k\le 1, \qW_k\succeq \qzero, \forall k.\notag
\eea

The constraint in \eqref{eqn:delta:h:SINR} is in general difficult to handle because the robust beamforming needs to satisfy infinite  number of channel realizations defined in the set $\mathcal{H}_k$. Fortunately, thanks to S-Procedure, it  can be equivalently written as the following SDP constraint \cite{GanRobust}
\be\label{eq:robust:SINR}
    \left[\begin{array}{ccc}
     \hat\qh_k^\dag \qQ_k \qh_k^\dag -  \gamma_k\left(N_0 + \frac{N_k}{\rho_k} \right) - u_k \sigma_k^2  & \qh_k^\dag\qQ_k     \\
      \qQ_k \qh_k & \qQ_k + u_k\qI
    \end{array}\right]\succeq \qzero, \exists u_k\ge 0.
\ee

Following the same principle, the constraint in \eqref{eqn:delta:h:EH} can be equivalently written as
\bea\label{eq:robust:EH}
    \left[\begin{array}{ccc}
     \hat\qh_k^\dag (\sum\limits_{j=1}^K  \qW_j) \qh_k^\dag -  \frac{\lambda_k}{1-\rho_k}+ N_0 - v_k \sigma_k^2 & \qh_k^\dag(\sum\limits_{j=1}^K  \qW_j)    \\
     (\sum\limits_{j=1}^K  \qW_j) \qh_k & (\sum\limits_{j=1}^K  \qW_j) + v_k\qI
    \end{array}\right]\\
    \succeq \qzero, \exists v_k \ge 0.
\eea

{{The worst-case SAR constraint \eqref{eqn:delta:Radiation} introduces further difficulty, and next we convert it to an equivalent convex constraint, and the result is summarized in the theorem below.}}
\begin{theorem}\label{theo3}
The worst-case SAR constraint \eqref{eqn:delta:Radiation} is equivalent to the constraint
\be
    \tr( \hat \qA_l (\sum_{j=1}^K  \qW_k) )+ \tau_l \|(\sum_{j=1}^K  \qW_k) \|_F  \leq   P_{l}, \forall l.\label{eq:SAR:worst}
\ee
\end{theorem}
The proof is given in Appendix B.

With the  robust constraints \eqref{eq:robust:SINR}, \eqref{eq:robust:EH} and \eqref{eq:SAR:worst}, we can obtain the equivalent robust beamforming optimization problem as follows
{\small
\bea\label{eqn:prob:4}
&&\textbf{P12:}~ \min_{\{\qW_k, \rho_k, u_k, v_k\}}~\sum_{k=1}^K  \tr(\qW_k) \\
&&\mbox{s.t.}~
           \left[\begin{array}{ccc}
     \hat\qh_k^\dag \qQ_k \qh_k^\dag -  \gamma_k\left(N_0 + \frac{N_k}{\rho_k} \right) - u_k \sigma_k^2  & \qh_k^\dag\qQ_k     \\
      \qQ_k \qh_k & \qQ_k + u_k\qI
    \end{array}\right]\notag\\
    && \succeq \qzero, \exists u_k,\forall k, \label{eq:robust:SINR2}\\
&&      \left[\begin{array}{ccc}
     \hat\qh_k^\dag (\sum\limits_{j=1}^K  \qW_j) \qh_k^\dag -  \frac{\lambda_k}{1-\rho_k}+ N_0 - v_k \sigma_k^2 & \qh_k^\dag(\sum\limits_{j=1}^K  \qW_j)  \notag  \\
     (\sum\limits_{j=1}^K  \qW_j) \qh_k & (\sum\limits_{j=1}^K  \qW_j) + v_k\qI
    \end{array}\right]\\
    && \succeq \qzero, \exists v_k, \forall k, \label{eq:robust:EH2}\\
&&  \tr( \hat \qA_l (\sum_{j=1}^K  \qW_k) )+ \tau_l \|(\sum_{j=1}^K  \qW_k) \|_F  \leq   P_{l}, \forall l,\notag\\
&& 0\le \rho_k\le 1, \qW_k\succeq \qzero, \forall k,\notag
\eea}
where we have ignored the rank constraint on $\qW_k$.

However, \textbf{P12} is  still a nonconvex problem because of the nonlinear terms about $\rho_k$ in the constraints \eqref{eq:robust:SINR2} and \eqref{eq:robust:EH2}. To tackle this new challenge,  we introduce the auxiliary variables $\{m_k, n_k\}$ and use the following formulation to convert it to a convex problem below
{\small \bea\label{eqn:prob:5}
&& \textbf{P13:}\min_{\{\qW_k, \rho_k, u_k, v_k, m_k, n_k\}}~\sum_{k=1}^K  \tr(\qW_k) \\
&& \mbox{s.t.}~
           \left[\begin{array}{ccc}
     \hat\qh_k^\dag \qQ_k \qh_k^\dag -  \gamma_k\left(N_0 + {N_C} m_k \right) - u_k \sigma_k^2  & \qh_k^\dag\qQ_k     \\
      \qQ_k \qh_k & \qQ_k + u_k\qI
    \end{array}\right]\notag\\
    && \succeq \qzero, \exists u_k, \forall k,\notag\\
&&      \left[\begin{array}{ccc}
     \hat\qh_k^\dag (\sum\limits_{j=1}^K  \qW_j) \qh_k^\dag -  \lambda_k n_k+ N_0 - v_k \sigma_k^2 & \qh_k^\dag(\sum\limits_{j=1}^K  \qW_j)    \\
     (\sum\limits_{j=1}^K  \qW_j) \qh_k & (\sum\limits_{j=1}^K  \qW_j) + v_k\qI
    \end{array}\right]\notag\\
    && \succeq \qzero, \exists v_k, \forall k,\notag\\
&& (1+ \frac{\tau_l }{\|\hat \qA_l \|_F}) \tr( \hat \qA_l (\sum_{j=1}^K  \qW_k) ) \leq   P_{l},  \notag\\
&& m_k\ge \frac{1}{\rho_k}, n_k \ge \frac{1}{1-\rho_k}, \forall k, \label{eq:mn} \ \\
&& 0\le \rho_k\le 1, \qW_k\succeq \qzero, u_k\ge 0, v_k\ge 0,\forall k.\notag
\eea}

Next we prove that the constraint \eqref{eq:mn} is tight at the optimum, therefore \textbf{P13} is equivalent to \textbf{P12}.
\begin{theorem}
There exists optimal solutions $\{m_k^*, n_k^*, \rho_k^*\}$ to \textbf{P13} such that  $m_k^*= \frac{1}{\rho_k^*}$ and  $n_k = \frac{1}{1-\rho_k^*}, \forall k$.
\end{theorem}
\begin{proof}
This can be proved by noting that given a solution  $\{m_k^*, n_k^*, \rho_k^*\}$, if any constraint  \eqref{eq:mn} is not tight, we can always reduce  $m_k^*$ and/or $n_k^*$ to make inequality constraints to be equalities without affecting the feasibility of other constraints. Therefore there always exists $\{m_k^*, n_k^*, \rho_k^*\}$ such that the constraint  \eqref{eq:mn} is tight. This completes the proof.
\end{proof}

Theorem 4 shows that introducing  the auxiliary variables $\{m_k, n_k\}$ in  \textbf{P13} to deal with the nonconvex constraints in  \textbf{P12} incurs no loss of optimality. However, since we have used SDP with relaxation in \textbf{P13}, only when the relaxation is tight, i.e., $\mbox{rand}(\qW_k)=1, \forall k$,   the optimal robust beamforming solution ${\qw_k}$  can be recovered from $\qW_k$; otherwise, the obtained solution provides a lower bound of the required transmit power. A feasible beamforming solution can be obtained by the standard randomization technique \cite{SDP}.

\section{Simulation Results}
Simulations   are carried out to evaluate the performance  of the proposed algorithms. We consider a  MISO downlink consisting of $K$ receivers randomly located around the BS with distance $l_k$ and direction $\zeta_k$ drawn from the uniform distribution, $l_k\sim  U(1,5)$m and $\zeta_i \sim  U(-\pi,\pi)$. Each receiver can harvest energy at frequency $f=915$ MHz and it is assumed that the antenna gains  at the BS and receivers are 8dBi and 3dBi, respectively. The path attenuation of receiver $k$, $L_k$, is obtained using the Friis equation with reference distance 1m and  path loss coefficient 2.5.  Because of the short distance between the BS and the receivers and dominance of the line-of-sight (LOS) signal, the Rician fading is used to model the   channel. Hence, $\mathbf{h}_k$ is composed of the  LOS  signal, $\mathbf{h}^{LOS}_k$ and the non-LOS signal  $\mathbf{h}^{NLOS}_k$ according to the expression below \cite{Zhang-14}
\begin{align}
	\mathbf{h}_k &= \sqrt{\frac{R}{1+R}} \mathbf{h}^{LOS}_k + \sqrt{\frac{1}{1+R}} \mathbf{h}^{NLOS}_k,
\end{align}
where $R=5$ dB is the Rician factor. For the LOS signal the far-field uniform linear antenna array model with $\lambda/2$ distance between antenna elements is considered \cite{karip2007} which implies that $\mathbf{h}^{LOS}_k = \sqrt{L_k}[1, \text{e}^{-j(1\pi\sin\zeta_k)},...,\text{e}^{-j((N-1)\pi\sin\zeta_k)}]^T$. Rayleigh fading is adopted for the NLOS signal, $\mathbf{h}^{NLOS}_k\in\mathbb{C}^{N_t\times 1}$  which means that each of its elements is a circularly symmetric complex Gaussian (CSCG) random variable with zero mean and variance $L_k$.

Unless otherwise specified, it is further assumed a BS with four antennas serving four receivers, i.e., $K=N_t=4$, $N_0=-70$ dBm and $N_C=-50$ dBm, while the  SINR and EH thresholds are the same for all receivers, i.e. $\bar\Gamma_k=\bar\Gamma=10$ dB, $\bar\lambda_k=\bar\lambda=-15$ dBm, $\forall k$. we assume that the total power constraint is $P_T=2$ W, and the SAR power constraint is $P_l=P=1.6$ W/kg, $\forall l$. {{There is no guarantee that the diode (rectification circuit) operates always in the transition regime, and the system could operate in the close to saturation regime, so it is important to capture saturation effects. }} We use the nonlinear energy harvesting model below proposed in \cite{Yunfei-17}
\be
    F_k(x)=\frac{\bar ax+\bar b}{x+c} - \frac{\bar b}{c}, \forall k, \label{eqn:nonlinear}
\ee
and the fitted parameters of the proposed model are $\bar a = 2.463, \bar b = 1.635, c = 0.826$ using the data in \cite{EH-data}. {\eqref{eqn:nonlinear}    belongs to the category of models which are based on real-world measurements by adjusting the parameters of a non-linear function  through curve fitting tools. It can model both saturation and non-saturation regimes, and has been widely used in the literature for the design of WPT/SWIPT systems such as\cite{Letaief-18,Kim-18,Tran-18}.}
{{Note that in general $F(\cdot)$ should be a function of the received signal rather than just its power, so {\eqref{eqn:nonlinear} is a simplified model  based on specific measurements under specific conditions and with continuous wave input signals.}} The SAR matrix is given by

\be
    \qA = \left[
            \begin{array}{cccc}
              1.6 & -1.2j & -0.42 & 0 \\
              1.2j & 1.6 & -1.2j & -0.42 \\
              -0.42  &  1.2j  & 1.6 & -1.2j \\
              0 & -0.42 & 1.2j &1.6 \\
            \end{array}
          \right],
\ee
for four antennas at the BS.

We use the minimum achievable SINR and EH ratio $t$ as the main performance evaluation criterion. The proposed optimal solutions, fixed beamforming solutions including ZF and RZF beamforming schemes as well as the hybrid beamforming scheme will be compared. In addition, the conventional backoff scheme is used as another performance benchmark described as follows. Firstly, it solves \textbf{P1} without the SAR constraint, i.e., the problem below
\bea\label{eqn:prob:01}
\textbf{P14:} \max_{\{\qw_k, \rho_k,t\}}&& t \\
\mbox{s.t.}&&
  \Gamma_m = \frac{ |\qh_{k}^\dag\qw_k|^2}{\sum\limits_{j=1,j \ne
k}^K  |\qh_{k}^\dag\qw_j|^2 + N_0 + \frac{N_k}{\rho_k}} \ge t \bar \gamma_k\triangleq \gamma_k,\notag\\
&& EH\left((1-\rho_k) \left(\sum\limits_{j=1}^K  |\qh_{i}^\dag\qw_k|^2 + N_0  \right)\right) \ge t \bar\lambda_k, \notag\\
&& 0\le \rho_k\le 1, \forall k, \notag \\
&& \sum_{k=1}^K \|\qw_k\|^2 \le P_T.
\eea
Suppose its solution is $\qw$. Define $\delta_l \triangleq \frac{\sum_{k=1}^K \qw_k^\dag \qA_{l} \qw_k}{P_{l}}$, then the backoff solution is given by  $\qw = \frac{\qw}{\min(1,\max_l \delta_l)}$ to satisfy the SAR   constraints.

 We first evaluate the performance and complexity of the fast algorithm for a single receiver system with a single SAR constraint ($K=L=1$) as the SAR limit varies in  Fig. \ref{fig:SU}. From Fig. \ref{fig:SU} (a), we can see that the minimum achievable  SINR and EH ratio of the proposed fast algorithm is identical to that of the SDP approach for systems with general number of receivers, which verifies its optimality. Fig. \ref{fig:SU} (b) shows that the computation time of the proposed fast algorithm is more than three orders of magnitude lower than the SDP approach. These results clearly demonstrate the advantage of the proposed fast algorithm.

Figs. \ref{fig:vs:radiation}, \ref{fig:vs:EH} and \ref{fig:vs:txpower}  depict the minimum achievable  SINR and EH ratio
by the different investigated algorithms by varying $P$, $\lambda$ and $P_T$, respectively.  Fig. \ref{fig:vs:radiation} (a) depicts that the proposed optimal solution can achieve substantial gain over the traditional backoff solution, especially when the SAR power constraint is low. As the SAR power constraint increases, the performance of the optimal solution is close to that without the SAR power constraint. The hybrid beamforming solution is superior to the two fixed beamforming solutions, while ZF beamforming performs much better than the RZF beamforming. Fig. \ref{fig:vs:radiation} (b) shows the feasibility of various schemes, which follows the similar trend as the results in  Fig. \ref{fig:vs:radiation} (a). In Fig. \ref{fig:vs:radiation} (c), we plot the actual transmit power consumed. It can be seen that without the SAR constraint, the total power budget of 2W is always used. The ZF beamforming solution uses similar power as the backoff solution, while  the hybrid beamforming scheme uses much less power, thus provides a higher power efficiency.

In Fig. \ref{fig:vs:EH}, we plot the minimum achievable  SINR and EH ratio  against the EH constraint. It is seen that the proposed optimal solution outperforms other schemes significantly. As expected, as the required EH increases, the performance  of all schemes degrades, and eventually converges to the similar result. The hybrid beamforming scheme achieves similar performance as the backoff scheme, but with much reduced complexity.

In Fig. \ref{fig:vs:txpower}, we examine the impact of the total transmit power on the minimum achievable  SINR and EH ratio.  It is observed that the performance of all schemes is increasing or non-decreasing as the total transmit power increases and saturates when the transmit power is greater than 34 dBm, except for the backoff solution. The reason is that although the transmit power increases, the SAR constraint becomes the bottleneck, so the increased power cannot be utilized and transferred to improved performance.  The backoff solution is penalized most in the high power regime, therefore its performance deteriorates quickly as the transmit power goes above 31 dBm.

Next, we investigate the impact of the channel and SAR matrix estimation error.  Fig. \ref{fig:vs:robust:average} depicts the minimum achievable  SINR and EH ratio of different schemes when the estimation error $\sigma^2_k=\tau_l, \forall k,l$ varies. The performance degradation is clearly seen as the estimation error grows. The performance of the proposed robust solution is still satisfactory when the estimation error is small, and it can achieve much higher     SINR and EH   than the non-robust solution that ignores the estimation error.   Fig. \ref{fig:vs:robust:cdf} (a) and (b)   further  depict the cumulative distribution functions (CDFs) of the output SINR and EH, respectively,  when there are estimation errors given by $\sigma^2_k=5\times 10^{-8}, \forall k$ and $\tau_l=7\times 10^{-8}, \forall l$. The vertical lines show the target SINR of 10 dB and the target EH of -15 dBm.  It is clearly seen that our proposed robust solution can guarantee the worst-case SINR and EH constraints are met even in the presence of estimation error. However, for the non-robust solution, both constraints are violated with probabilities higher than 90\%.

     \begin{figure}[]
  \centering
  \subfigure[]{
       \includegraphics[width=0.35\textwidth]{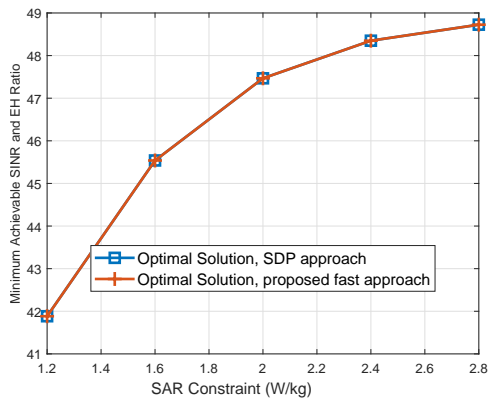}
     }
  \subfigure[]{
     \includegraphics[width=0.35\textwidth]{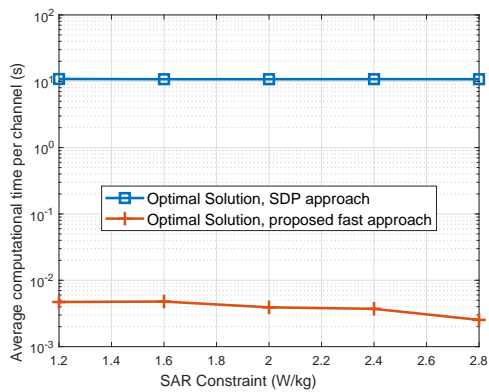}
    }

  \caption{Single-receiver performance: (a) the  minimum achievable SINR and EH ratio   and (b) the computational time.}
  \label{fig:SU}
\end{figure}

    \begin{figure}[]
  \centering
       \subfigure[]{  \includegraphics[width=0.35\textwidth]{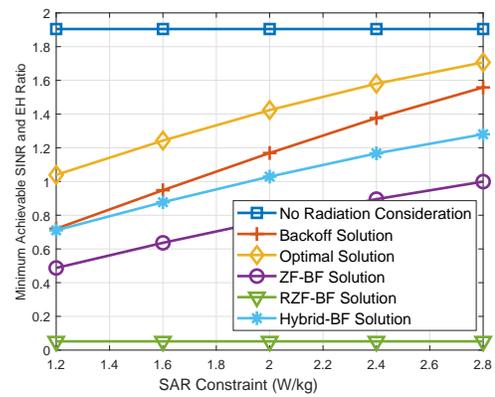}}
          \subfigure[]{\includegraphics[width=0.35\textwidth]{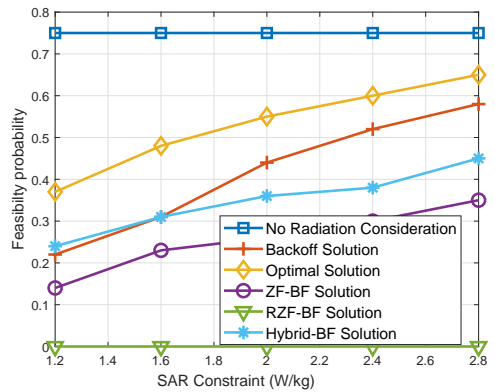}}
            \subfigure[]{\includegraphics[width=0.35\textwidth]{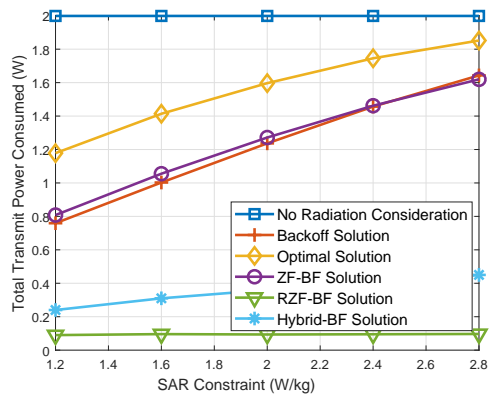}
    }
   \caption{The minimum achievable  SINR and EH ratio  vs SAR: (a) the minimum achievable  SINR and EH ratio, (b) the feasibility probability and c) total transmit power consumed.}
  \label{fig:vs:radiation}
\end{figure}

    \begin{figure}[]
  \centering
       \includegraphics[width=0.35\textwidth]{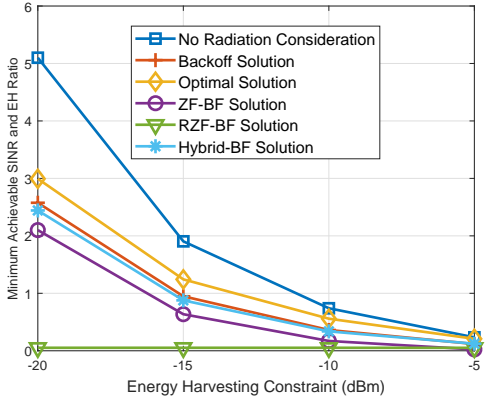}
   \caption{The minimum achievable  SINR and EH ratio vs the EH constraint.}
  \label{fig:vs:EH}
\end{figure}

    \begin{figure}[]
  \centering
      \includegraphics[width=0.35\textwidth]{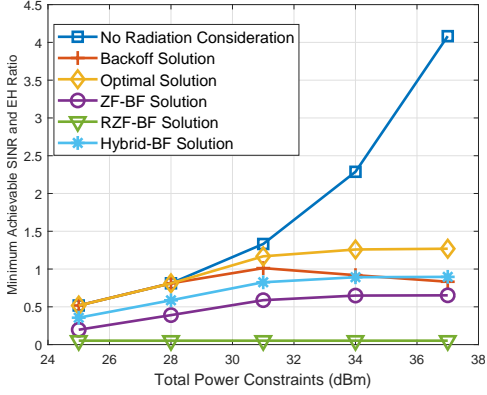}
   \caption{The minimum achievable  SINR and EH ratio vs the total transmit power.}
  \label{fig:vs:txpower}
\end{figure}

    \begin{figure}[]
  \centering
      \includegraphics[width=0.35\textwidth]{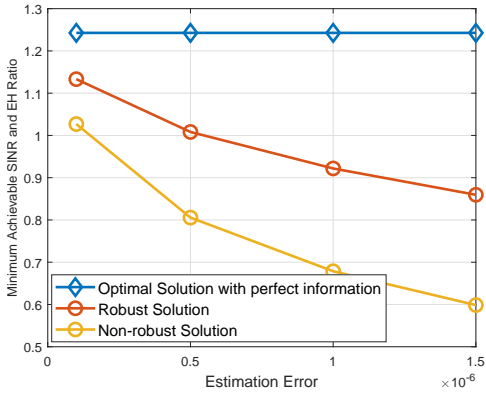}
   \caption{The minimum achievable  SINR and EH ratio vs the estimation error.}
  \label{fig:vs:robust:average}
\end{figure}

    \begin{figure}[]
  \centering
  \subfigure[]{
       \includegraphics[width=0.35\textwidth]{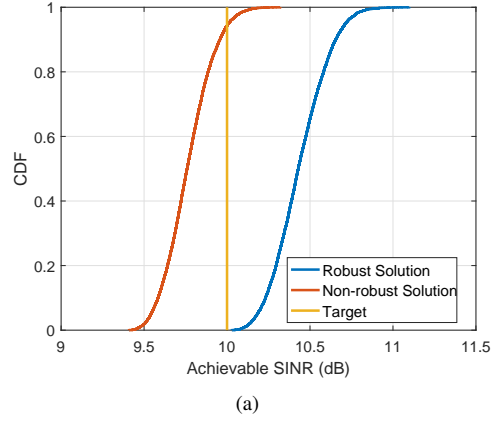}
     }
  \subfigure[]{
     \includegraphics[width=0.4\textwidth]{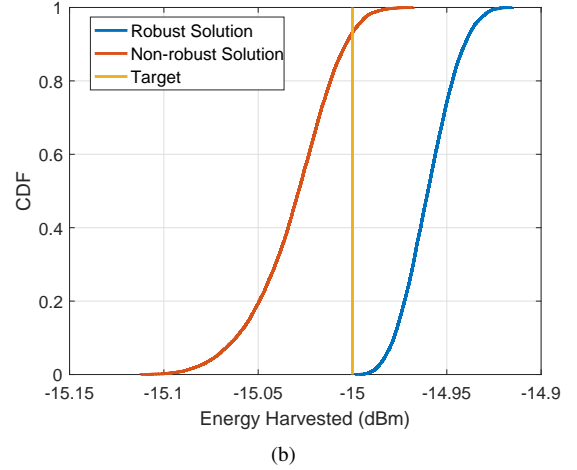}
    }
  \caption{The minimum achievable  SINR and EH ratio with estimation error: (a) received SINR   and (b) energy harvested.}
  \label{fig:vs:robust:cdf}
\end{figure}

\section{Conclusions}
In this paper, we have studied the optimization of SAR-constrained multiuser transmit beamforming of a  SWIPT system. In the general case with perfect information, we have shown that  the optimal beamforming and power splitting solutions can be obtained via semidefinite programming and bisection search; while a much more efficient solution can be found for the special single-receiver single-SAR case. We further designed low-complexity suboptimal solutions including the fixed beamforming and hybrid beamforming schemes. In addition, we proposed robust beamforming solutions to deal with the imperfect channel and SAR matrix information, while guaranteeing the required SINR, EH constraints and the maximum SAR is below the given threshold.  Our simulation and analysis have shown significant performance improvement of the proposed SAR-aware optimal solution  over the  conventional transmission scheme with simple power backoff. Future works include using large intelligence surface to further reduce the energy consumption of the transmitter while satisfying the SAR constraints.

\section*{Appendix A. Proof of Theorem \ref{theo2}}
\begin{proof}

The partial Lagrangian of the problem \textbf{P5} is
{\small
\bea
&& L(\{\qW_k, \rho_k,\alpha_k, \beta_k, \nu_l\}) \notag\\
 && = \sum_{k=1}^K  \tr(\qW_k)+ \sum_{l=1}^L \nu_l \left(\sum_{k=1}^K \tr( \qA_l \qW_k) - P_{l} \right) \notag\\
  && + \sum_{k=1}^K\alpha_k \left( \sum\limits_{j=1}^K  \tr(\qh_{k}\qh_{k}^\dag\qW_j) + N_0 + \frac{N_k}{\rho_k} - (1+\frac{1}{\gamma_k})  \tr(\qh_{k}\qh_{k}^\dag\qW_k)\right) \notag \\
   && + \sum_{k=1}^K\beta_k \left(\lambda_k - (1-\rho_k) \left(\sum\limits_{j=1}^K  \tr(\qh_{k}\qh_{k}^\dag\qW_j) + N_0  \right)\right)\notag\\
   && = \sum_{k=1}^K \tr(\qW_k \qX_k) - \sum_{l=1}^L \nu_l P_l  + \sum_{k=1}^K\alpha_k ( N_0 + \frac{N_i}{\rho_i} )\notag\\
    &&+ \sum_{k=1}^K\beta_k ( \lambda_k- (1-\rho_k)N_0),\notag
\eea}
where $\{\alpha_k, \beta_k, \nu_l\}$ are dual variables, and we have defined
{\small
\bea
    \qX_k   \triangleq \qI - \sum_{j=1}^K\beta_j(1-\rho_j)  \qh_{j}\qh_{j}^\dag +   \sum_{j=1}^K\alpha_j   \qh_{j}\qh_{j}^\dag
    - \alpha_k (1+\frac{1}{\gamma_k}) \qh_{k}\qh_{k}^\dag
    + \sum_{l=1}^L \nu_l \qA_l.
\eea}

So the dual problem is
{ \bea
    &&\max_{\bm\nu, \bm\alpha, \bm\beta\ge 0} ~ - \sum_{l=1}^L \nu_l P_l  + \sum_{k=1}^K\alpha_k ( N_0 + \frac{N_i}{\rho_i} )\notag\\
     &&+ \sum_{k=1}^K\beta_k ( \lambda_k- (1-\rho_k)N_0)\notag\\
   && \mbox{s.t.} ~ \qX_k   = \qI + \sum_{j=1}^K (\alpha_j-\beta_j(1-\rho_j))    \qh_{j}\qh_{j}^\dag \notag\\
    &&    - \alpha_k (1+\frac{1}{\gamma_k}) \qh_{k}\qh_{k}^\dag
    + \sum_{l=1}^L \nu_l \qA_l\succeq \qzero, \forall k.\notag
\eea}

Next we prove that the matrix $\qI + \sum_{j=1}^K (\alpha_j-\beta_j(1-\rho_j))    \qh_{j}\qh_{j}^\dag + \sum_{l=1}^L \nu_l \qA_l$ is full rank by contradiction. If it is not full-rank, suppose there exits a non-zero vector $\qx$ that satisfies $\qx^\dag(\qI + \sum_{j=1}^K (\alpha_j-\beta_j(1-\rho_j))    \qh_{j}\qh_{j}^\dag + \sum_{l=1}^L \nu_l \qA_l)\qx=0$.
Because $\qX_k\succeq \qzero$, we have
\bea
    \qx^\dag\qX_k\qx &=& \qx^\dag ( \qI + \sum_{j=1}^K (\alpha_j-\beta_j(1-\rho_j))    \qh_{j}\qh_{j}^\dag
    + \sum_{l=1}^L \nu_l \qA_l )\qx \notag\\
    && - \qx^\dag(\alpha_k (1+\frac{1}{\gamma_k}) \qh_{k}\qh_{k}^\dag) \qx\notag\\
    &=& -  \alpha_k (1+\frac{1}{\gamma_k}) |\qh_{k}^\dag\qx|^2   \ge 0.
\eea
Therefore it holds true that
\be
    \qh_{k}^\dag  \qx=0, \forall k.
\ee

It follows that
\bea
    \qx^\dag(\qI + \sum_{j=1}^K (\alpha_j-\beta_j(1-\rho_j))    \qh_{j}\qh_{j}^\dag + \sum_{l=1}^L \nu_l \qA_l)\qx\notag\\
     =     \qx^\dag(\qI +   \sum_{l=1}^L \nu_l \qA_l)\qx>0,
\eea
which contradicts the assumption that $\qx^\dag(\qI + \sum_{j=1}^K (\alpha_j-\beta_j(1-\rho_j))    \qh_{j}\qh_{j}^\dag + \sum_{l=1}^L \nu_l \qA_l)\qx=0$. Therefore the matrix $\qI + \sum_{j=1}^K (\alpha_j-\beta_j(1-\rho_j))    \qh_{j}\qh_{j}^\dag + \sum_{l=1}^L \nu_l \qA_l$ must be full rank, and the rank of $\qX_k$ is at least $N_t-1$.

One  Karush–Kuhn–Tucker (KKT) condition of the problem \textbf{P5} is that $\mbox{trace}(\qW_k \qX_k)=0$, so the rank of $\qW_k$ is at most 1. This completes the proof.

\end{proof}
\section*{Appendix B. Proof of Theorem \ref{theo3}}
\begin{proof}
We consider the following optimization problem:
\bea\label{eqn:robust:SAR}
\max_{\qX} && \tr(\bar \qW \qX) ~~~~\mbox{s.t.}~~ \|\qX\|_F\le \tau,
\eea
where $\qX$ has the meaning of the worst SAR error matrix. Its Lagrangian is given by
\bea
    L &=& -\tr(\bar \qW \qX) + u (\|\qX\|_F^2 -  \tau^2) \notag\\
    &=&  -\tr(\bar \qW \qX) + u (\tr(\qX\qX^\dag) -  \tau^2),
\eea
where $u$ is the dual variable. Setting its first-order derivative to be zero leads to:
\be
    2u \qX - \bar \qW=0.
\ee
Therefore $\qX$ should be a scaled version of the $\bar \qW$. It is easy to see that  the norm constraint in \eqref{eqn:robust:SAR} should be tight, so we can obtain the worst-case $\qX$ as $\qX^* = \frac{\bar \qW}{\|\bar \qW\|_F}\tau$, and $\max_\qX \tr(\bar \qW \qX) =  \tau\|\bar \qW\|_F$. Substituting it into \eqref{eqn:delta:Radiation} leads to the constraint \eqref{eq:SAR:worst}. This completes the proof.
\end{proof}


\vspace{-3mm}

\begin{thebibliography}{10}
\providecommand{\url}[1]{#1}
\csname url@samestyle\endcsname
\providecommand{\newblock}{\relax}
\providecommand{\bibinfo}[2]{#2}
\providecommand{\BIBentrySTDinterwordspacing}{\spaceskip=0pt\relax}
\providecommand{\BIBentryALTinterwordstretchfactor}{4}
\providecommand{\BIBentryALTinterwordspacing}{\spaceskip=\fontdimen2\font plus
\BIBentryALTinterwordstretchfactor\fontdimen3\font minus
  \fontdimen4\font\relax}
\providecommand{\BIBforeignlanguage}[2]{{%
\expandafter\ifx\csname l@#1\endcsname\relax
\typeout{** WARNING: IEEEtran.bst: No hyphenation pattern has been}%
\typeout{** loaded for the language `#1'. Using the pattern for}%
\typeout{** the default language instead.}%
\else
\language=\csname l@#1\endcsname
\fi
#2}}
\providecommand{\BIBdecl}{\relax}
\BIBdecl

\bibitem{Ioannis-ComMag-14} I. Krikidis, S. Timotheou, S. Nikolaou, G. Zheng, D. W. K. Ng, and
R. Schober, ``Simultaneous wireless information and power transfer in
modern communication systems,'' {\em IEEE Commun. Mag.}, vol. 52, no. 11,
pp. 104--110, Nov. 2014.

\bibitem{CLE} B. Clerckx, R. Zhang, R. Schober, D. W. K. Ng, D. I. Kim, and H. V. Poor, ``Fundamentals of wireless information and power transfer: From RF energy harvester models to signal and system designs,'' {\it IEEE J. Selec. Areas Commun.}, vol. 37, no. 1, pp. 4--33, Jan. 2019.

\bibitem{VAR} L. R. Varshney, ``Transporting information and energy simultaneously,'' in {\it Proc. IEEE Int. Symp. Inf. Theory}, Toronto, Canada, Jul. 2008, pp. 1612--1616.

\bibitem{PER} S. B. Amor, S. M. Perlaza, I. Krikidis, and H. V. Poor, ``Feedback enhances simultaneous wireless information and  energy transmission in multiple access channels,'' {\it IEEE Trans. Inf. Theory}, vol. 63, no. 8, pp. 5244--5265, Aug. 2017.

\bibitem{ZHA}  R. Zhang and C. K. Ho, ``MIMO broadcasting for simultaneous wireless information and power transfer,'' {\it IEEE Trans. Wireless Commun.}, vol. 12, no. 5, pp. 1989--2001, May 2013.

\bibitem{Stelios-TWC-14}
    S. Timotheou, I. Krikidis, G. Zheng, and  B. Ottersten, ``Beamforming for MISO interference channels with QoS and RF energy transfer,'' {\em IEEE
Trans. Wireless Commun.}, vol. 13, no. 5, pp. 2646–-2658, May 2014.

  \bibitem{Zhang-14}
Q. Shi, L. Liu, W. Xu, and R. Zhang, ``Joint transmit beamforming
and receive power splitting for MISO SWIPT systems,'' {\em IEEE
Trans. Wireless Commun.}, vol. 13, no. 6, pp. 3269--3280, Jun.
2014.


\bibitem{RUI1} X. Zhou, R. Zhang, and C. K. Ho, ``Wireless information and power transfer: Architecture design and rate-energy tradeoff,'' {\it IEEE Trans. Commun.}, vol. 61, no. 11,  pp. 4754--4767, Nov. 2013.


\bibitem{ZHA2} R. Zhang, L. -L. Yang, and L. Hanzo, ``Energy pattern aided simultaneous wireless information and power transfer,'' {\it IEEE J. Selec. Areas Commun.}, vol. 33, no. 8,  pp. 1492--1505, Aug. 2015.


\bibitem{CLE3} B. Clerckx and E. Bayguzina, ``Waveform design for wireless power transfer,'' {\it IEEE Trans. Signal Process.}, vol. 64, no. 23,  pp. 6313--6328, Dec. 2016.

\bibitem{CLE4} B. Clreckx, ``Wireless information and power transfer: nonlinearity, waveform design, and rate-energy tradeoff,'' {\it IEEE Trans. Signal Process.}, vol. 66, no. 4, pp. 847--862, Feb. 2018.

\bibitem{VOL} N. Volkow, D. Tomasi, G. -J. Wang, P. Vaska, J. Fowler, F. Telang, D. Alexoff, J. Logan, and C. Wong, ``Effects of dell phone radiofrequency signal exposure on brain glucose metabolism,'' {\it J. Amer. Med. Assoc.}, vol. 305, pp. 808--813, Feb. 2011.

\bibitem{IARC}  ``ARC classifies radiofrequency electromagnetic fields as
possibly carcinogenic to humans,'' Press Release No.
208, International Agency for Research on Cancer, World
Health Organization, May 2011.


\bibitem{safe_charging}
 H. Dai, Y. Liu, G. Chen, X. Wu, T. He, and A. X. Liu, ``Safe Charging for Wireless Power Transfer,'' {\em IEEE/ACM Trans. Netw}, vol. 25, no. 6, pp. 3531-3544, Dec. 2017.


\bibitem{safe_distributed}
K. S. Yildirim, R. Carli and L. Schenato, ``Safe distributed control of wireless power transfer networks,'' {\em  IEEE Internet Things J.}, vol. 6, no. 1, Feb. 2019.


\bibitem{scheduling_charging}
H. Dai, H. Ma,  A. X. Liu, and G. Chen,  ``Radiation constrained scheduling of wireless charging tasks,''  {\em IEEE/ACM Trans. Netw}, vol. 26, no. 1, pp. 314-327, Feb. 2018.

\bibitem{secure_charging}
Q. Liu, K. S. Yildirim, P. Pawełczak, and M. Warnier, ``Safe and secure wireless power transfer networks: challenges and opportunities in RF-based systems,''
{\em  IEEE Commun. Mag.}, vol. 54, no. 9, pp. 74-79, Sept. 2016.



\bibitem{FCC-01}
K. Chan and R. F. Cleveland, and D. L. Means ``Evaluating compliance with FCC
guidelines for human exposure to radiofrequency electromagnetic fields,''
{\em Federal Commun. Commiss.}, Washington, DC, USA, Tech. Rep. 01--01,
Jun. 2001.


\bibitem{ULC} CENELEC ES 59005, ``Considerations for evaluation of human exposure to ElectromagneticFields (EMFs) from Mobile Telecommunication Equipment (MTE) in the frequency range 30MHz – 6 GHz'', {\em European Specification, European Committee for Electrotechnical Standardization (CENELEC)}, Oct. 1998.

\bibitem{iphone} ``iPhone 6 RF Exposure information,''  https://www.apple.com/legal/rfexposure/iphone7,2/en/.

\bibitem{Murch-04}
 K.-C. Chim, K. C. Chan, and R. D. Murch, ``Investigating the impact of
smart antennas on SAR,'' {\em IEEE Trans. Antennas Propag.},  vol. 52, no. 5, pp. 1370–1374, May 2004.

\bibitem{Qiang-07}
 C. Qiang, Y. Komukai, and K. Sawaya, ``SAR investigation of array
antennas for mobile handsets,'' {\em IEICE Trans. Commun.}, vol. 90, no. 6, pp. 1354–1356, Jun. 2007.

\bibitem{Mahmoud-08}
K. R. Mahmoud, M. El-Adawy, S. M. Ibrahem, R. Bansal, and S. H.
Zainud-Deen, ``Investigating the interaction between a human head and
a smart handset for 4G mobile communication systems,'' {\em Progress In
Electromagnetics Research C}, vol. 2, pp. 169–188, 2008.

\bibitem{Hochwald-12}
B. M. Hochwald and D. J. Love, ``Minimizing exposure to electromagnetic radiation in portable devices,'' in {\em Proc. Inf. Theory Appl.
Workshop (ITA)}, San Diego, CA, USA, 2012, pp. 255--261.

\bibitem{Hochwald-14}
B. M. Hochwald, D. J. Love, S. Yan, P. Fay, and J.-M. Jin, ``Incorporating
specific absorption rate constraints into wireless signal design,'' {\em IEEE
Commun. Mag.}, vol. 52, no. 9, pp. 126--133, Sep. 2014.




\bibitem{Ying-CISS-13}
D. Ying, D. J. Love, and B. M. Hochwald, ``Beamformer optimization
with a constraint on user electromagnetic radiation exposure,'' in {\em Proc.
47th Annu. Conf. Inf. Sci. Syst. (CISS)}, Baltimore, MD, USA, 2013,
pp. 1--6.

\bibitem{Ying-Globecom-13}
 D. Ying, D. J. Love, and B. M. Hochwald, ``Transmit covariance optimization with a constraint on user electromagnetic radiation exposure,''
in {\em Proc. Global Commun. Conf. (GLOBECOM)}, Atlanta, GA, USA,
Dec. 2013, pp. 4104--4109.

\bibitem{Ying-TWC-13}
  D. Ying, D. J. Love, and B. M. Hochwald, ``Closed-loop precoding
and capacity analysis for multiple antenna wireless systems with user
radiation exposure constraints,'' {\em IEEE Trans. Wireless Commun.}, vol. 14,
no. 10, pp. 5859--5870, Oct. 2015.


\bibitem{Ying-TWC-17}
D. Ying, D. J. Love, and B. M. Hochwald, ``Sum-rate analysis for multi-user MIMO systems with user exposure constraints,'' {\em IEEE Trans. Wireless Commun.}, vol. 16,
no. 11, pp. 7376--7388, Nov. 2017.

\bibitem{Rui-TIT-12}
L. Zhang, R. Zhang, Y. C. Liang, Y. Xin, and H. V. Poor, ``On the Gaussian MIMO BC-MAC duality with multiple transmit covariance constraints,'' {\em IEEE Trans. Inf. Theory}, vol. 58, no. 34, pp. 2064--2078, Apr. 2012.

\bibitem{Rui-TCOM-17}
Y. Zeng, B. Clerckx, and R. Zhang, ``Communications and signals design for wireless power transmission,'' {\em IEEE Trans. Commun.}, vol. 65, no. 5, pp. 2264--2290, May 2017.

\bibitem{WANG} L. Wang, F. Hu, Z. Ling, and B. Wang, ``Wireless information and power transfer to maximize information throughput in WBAN,'' {\it IEEE Int. Things J.}, vol. 4, pp. 1663--1670, Oct. 2017.


\bibitem{XZHA} X. Zhang, K. Liu, and L. Tao, ``A cooperative communication scheme for full-duplex simultaneous wireless information and power transfer body area networks,'' {\it IEEE Sensor Lett.}, vol. 2, no. 4,   Dec. 2018.

\bibitem{ASIF} S. M. Asif, A. Iftikhar, J. W. Hansen, M. S. Khan, D. L. Ewert, and B. D. Braaten, ``A novel RF-powered wireless pacing via a rectenna-based pacemaker and a wearable transmit-antenna array,'' {\it IEEE Access}, vol. 7, pp. 1139--1148, Jan. 2019.

\bibitem{learning_signal}
M. Varasteh, E. Piovano, and B. Clerckx, ``A learning approach to
wireless information and power transfer signal and system design,'' {\em IEEE International Conference on Acoustics, Speech and Signal Processing (ICASSP)}, Brighton, UK, May 2019, pp. 4534–4538.

\bibitem{VAR} M. Varasteh, B. Rassouli, and B. Clerckx, ``Wireless information and power transfer over an AWGN: Nonlinearity and asymmetric Gaussian signalling,'' in {\it Proc. IEEE Inf. Theory Work.}, Kaohsiung, Taiwan, Nov. 2017, pp. 181--185.



\bibitem{experiment}
J. Kim, B. Clerckx, and P. D. Mitcheson, ``Signal and system design for wireless power transfer: prototype, experiment and validation,'' available at https://arxiv.org/abs/1901.01156.

\bibitem{tsps}
J.-M. Kang, I.-M. Kim, and D. I. Kim, ``Wireless Information and Power Transfer: Rate-Energy Tradeoff for Nonlinear Energy Harvesting,'' {\em IEEE Trans. Wireless Commun.}, vol. 17, no.3, pp. 1966-1981, Mar. 2018.

\bibitem{Li-robust-18}
S. Li, C. Li, S. Jin, M. Wei, and L. Yang, ``SINR balancing technique for robust beamforming in V2X-SWIPT system based on a non-linear EH model,'' {\em Physical Communication}, vol. 29, pp. 95--102, Aug.  2018.


\bibitem{Schober-15}
E. Boshkovska, D. W. K. Ng, N. Zlatanov, and R. Schober, ``Practical non-linear energy harvesting model and resource allocation for {SWIPT} systems,'' {\em IEEE Commun. Lett.}, vol. 19, no. 12, pp. 2082--2085, Dec. 2015.

\bibitem{Schober-17}
E. Boshkovska, D.W.K. Ng, N. Zlatanov, A. Koelpin, R. Schober, ``Robust resource allocation for {MIMO} wireless powered communication networks based on a
non-linear {EH} model'', {\em IEEE Trans. Commun.}, vol. 65, no. 5, pp. 1984–1999, May 2017.

\bibitem{Yunfei-17}
Y. Chen, N. Zhao, M.-S. Alouini, ``Wireless energy harvesting using signals from multiple fading channels,'' {\em IEEE Trans. Commun.}, vol. 65, no. 11, pp. 5027--5039, Nov. 2017.

\bibitem{Dong-16}
Y. Dong, M. J. Hossain, and J. Cheng, ``Performance of wireless powered amplify and forward relaying over nakagami-mfading channels with nonlinear energy harvester,'' {\em IEEE Commun. Lett.}, vol. 20, no. 4, pp. 672--675, Apr. 2016.


\bibitem{Xu-17}
X. Xu, A. zelikkale, T. McKelvey, and M. Viberg, ``Simultaneous information and power transfer under a non-linear RF energy harvesting model,'' in {\em Proc. IEEE Int. Conf. Commun. (ICC) Workshop}, May 2017, Paris, France, pp. 179--184.

\bibitem{Yunfei-16}
Y. Chen, K. T. Sabnis, and R. A. Abd-Alhameed, ``New formula for conversion efficiency of RF EH and its wireless applications,'' {\em IEEE Trans.  Veh. Tech.}, vol. 65, no. 11, pp. 9410-–9414, Nov. 2016.


\bibitem{EH-data}
T. Le, K. Mayaram, and T. Fiez, ``Efficient far-field radio frequency energy harvesting for passively powered sensor networks,'' {\em IEEE
J. Solid-State Circuits}, vol. 43, no. 5, pp. 1287--1302, May 2008.


\bibitem{Letaief-18}
Y. Lu, K. Xiong , P. Fan , Z. Zhong, and K. B. Letaief ``Robust transmit beamforming with artificial
redundant signals for secure swipt system under non-linear EH model,'' {\em IEEE Trans. Commun.}, vol. 17, no. 4, pp. 2218--2232, Apr. 2018.

\bibitem{Kim-18}
J.-M. Kang, I.-M. Kim, and D. I. Kim, ``Wireless information and power transfer: rate-energy tradeoff for nonlinear energy harvesting,'' {\em IEEE Trans. Wireless Commun.}, vol.17, pp. 1966-1981, Mar. 2018.

\bibitem{Tran-18}
H. Tran, G. Kaddoum, K. T. Truong, ``Resource allocation in swipt networks under a nonlinear energy harvesting model: power efficiency, user fairness, and channel nonreciprocity,'' {\em IEEE Trans.  Veh. Tech.}, vol. 67, no. 9,  pp. 8466--8480, Sept. 2018.

\bibitem{karip2007}
E. Karipidis, N. D. Sidiropoulos, and Z. Q. Luo, ``Far-field multicast beamforming for uniform linear antenna arrays,'' \emph{IEEE Trans. Signal Process.}, vol. 55, no. 10, pp. 4916--4927, 2007.
 \bibitem{RZF}
  C. Peel, B. Hochwald, and A. Swindlehurst, ``A vector-perturbation technique for near-capacity multiantenna multiuser communication-part
I: Channel inversion and regularization,'' {\em IEEE Trans. Commun.}, vol. 53, no. 1, pp. 195--202, Jan. 2005.


\bibitem{cvx}
M. Grant and S. Boyd. CVX: Matlab software for disciplined convex
programming, version 2.0 beta, Sept. 2012. Available: http://cvxr.com/cvx.

\bibitem{Nemirovski}
B.-T. Aharon and A. Nemirovski, {\em Lectures on Modern Convex Optimization: Analysis, Algorithms, and Engineering Applications}, MOS-SIAM
Series on Optimization, 2001.

\bibitem{GanRobust}
G. Zheng, K. K. Wong, and T. S. Ng, ``Robust linear MIMO in the downlink: A worst-case optimization with ellipsoidal uncertainty regions,'' \textit{EURASIP Journal on Advances in Signal Processing}, Dec. 2008,  2008:609028.

\bibitem{SDP}
Z.-Q. Luo, W.-K. Ma, A. M.-C. So, Y. Ye, and S. Zhang, ``Semidefinite relaxation of quadratic optimization problems: From its practical deployments and scope of applicability to key theoretical results,'' {\em IEEE Signal Process. Mag.}, vol. 27, no. 3, pp. 20-34, May 2010.

\end{thebibliography}
\end{document}